\documentclass{article}

\usepackage[letterpaper,top=2cm,bottom=2cm,left=3cm,right=3cm,marginparwidth=1.5cm]{geometry}

\usepackage{multirow} 
\usepackage[english]{babel}
\usepackage{graphicx}
\usepackage{amsmath,amssymb,bm, amsthm}
\usepackage{fancyhdr}
\usepackage{pdfpages}
\usepackage{setspace}
\usepackage[justification=justified]{caption}
\usepackage[toc,page,titletoc]{appendix}
\usepackage{csquotes}
\usepackage{caption}
\usepackage{float}
\usepackage{times}
\usepackage{bm}
\usepackage{color}
\usepackage{tabularx}
\usepackage{cases}
\usepackage{tabularx}
\usepackage[affil-it]{authblk}
\usepackage{subcaption}
\usepackage{makecell}
\usepackage{booktabs}
\usepackage{ulem}
\usepackage{caption}
\graphicspath{{figures}}

\usepackage{natbib}
\usepackage{doi}

\newtheorem{theorem}{Theorem}
\newtheorem{proposition}{Proposition}
\newtheorem{lemma}{Lemma}

\title{Repulsive mixtures via the sparsity-inducing partition prior}

\author[1, 4]{Alexander Mozdzen\thanks{Corresponding author at: Department of Paediatrics, National University of Singapore, NUHS Tower Block, Level 12, 1E Kent Ridge Road, Singapore 119228. 
Email: alex.mozdzen@gmail.com}}
\author[1]{Timothy Wertz}
\author[1,4]{Maria De Iorio}
\author[3]{Andrea Cremaschi}
\author[2]{Gregor Kastner}
\author[1,4]{Johan Eriksson}

\affil[1]{National University of Singapore, Singapore}
\affil[2]{University of Klagenfurt, Austria}
\affil[3]{IE University, Spain}
\affil[4]{Singapore Institute for Human Development and Potential, A*STAR, Singapore}

\date{\today}

\begin{document}

\renewcommand{\doitext}{DOI: }
\renewcommand{\url}[1]{}
\newcommand{\TW}[1]{{\color{teal}[TW: \textbf{#1}]}} 

\newcommand{\AM}[1]{{\color{purple}[AM: \textbf{#1}]}} 
\newcommand{\AMout}[1]{\textcolor{purple}{\sout{#1}}}

\newcommand{\MDI}[1]{\textcolor{magenta}{[MDI: \textbf{#1}]}}

\newcommand{\AC}[1]{\textcolor{blue}{[AC: \textbf{#1}]}}
\newcommand{\ACout}[1]{\textcolor{blue}{\sout{#1}}}

\newcommand{\GK}[1]{\textcolor{brown}{[GK: \textbf{#1}]}}
\newcommand{\GKout}[1]{\textcolor{brown}{\sout{#1}]}}

\newcommand{\iid}{\stackrel{\rm iid}{\sim}}
\newcommand{\ind}{\stackrel{\rm ind}{\sim}}

\newcommand{\Gammad}[1]{ \mathcal{G}\left(#1\right)}
\newcommand{\Gammainv}[1]{\mathcal{IG} \left(#1\right)}

\newcommand{\inwish}[1]{\mathcal{IW} \left(#1\right)}

\newcommand{\SDir}[1]{Selberg Dirichlet}
\newcommand{\Dir}[1]{Dirichlet}
\newcommand{\GSDir}[1]{generalized Selberg Dirichlet}

\newcommand{\lognorm}[1]{log\mathcal{N} \left(#1\right)}
\newcommand{\Mna}[1]{ M_{na}}

\newcommand{\phia}[1]{ \phi_{a}}
\newcommand{\phina}[1]{ \phi_{na}}

\newtheorem{definition}{Definition}


\newcommand{\alttext}[1]{}

\maketitle

\begin{abstract}

We introduce a novel prior distribution for modelling the weights in mixture models based on a generalisation of the  Dirichlet distribution, the Selberg Dirichlet distribution. This distribution contains a repulsive term, which naturally penalises values that lie close to each other on the simplex, thus encouraging few dominating clusters. The repulsive behaviour induces additional sparsity on the number of components. We refer to this construction as sparsity-inducing partition (SIP) prior. By highlighting differences with the conventional Dirichlet distribution, we present relevant properties of the SIP prior and demonstrate their implications across a variety of mixture models, including finite mixtures with a fixed or random number of components, as well as repulsive mixtures. We propose an efficient posterior sampling algorithm and validate our model through an extensive simulation study as well as an application to a biomedical dataset describing children's Body Mass Index and eating behaviour.

\end{abstract}

\section{Introduction}
Clustering is a fundamental statistical technique used to uncover underlying structures within datasets. Among the various approaches, mixture models have gained particular prominence owing to their flexibility in representing data as originating from multiple latent sub-populations \citep[see, for example,][]{ banfield_1993_modelbased_clust, bensmail_1997_modbased_clust,mclachlan_2000_fmm, frühwirth_schnatter_2006_fmm}. In a mixture model, observations are assumed to arise from one of
$M$ groups, each group being typically modelled by a distribution from a parametric family. The contribution of each group is referred to as a component of the mixture, and is weighted by the relative frequency of the group in the population. This provides a conceptually simple way of relaxing distributional assumptions and flexibly approximating distributions poorly captured by standard parametric families. Moreover, it provides a framework by which observations may be clustered together for purposes of discrimination or classification. The distribution of a set of $N$ observations $\bm{y} = (\bm y_1, \dots, \bm y_N)$ in a finite mixture model with $M$ components is:
\begin{align}
\label{eq:lik}
f_y( \bm{y} \mid M, \bm{w}, \bm{\theta} ) = \prod_{i=1}^N \sum_{m=1}^M w_m f\left( \bm y_i \mid \bm \theta_m \right)
\end{align}
where $\bm{w} = (w_1,\dots,w_M)$ denotes the weights, with $\sum_{m=1}^M w_m = 1$, $0 \leq w_m \leq 1$, while $\bm \theta = \left(\bm \theta_1, \dots, \bm \theta_M \right)$ denotes the array of component-specific parameter vectors.  In representing underlying clusters as components of the mixture, information about cluster sizes, shapes and locations can be determined through the estimation of $\bm{w}$ and $\bm \theta$.

In this work, we focus on the prior distribution for the mixture weights exploring an extension of the Dirichlet distribution, the \textit{Selberg Dirichlet} distribution. This prior influences the distribution of the number of clusters in the sample by incorporating a repulsive term to encourage dissimilarity among the weights, thereby inducing sparsity. We refer to this as the sparsity-inducing partition (SIP) prior. While a larger number of mixture components can potentially improve density estimation, uncovering underlying clusters is a nuanced problem without a ready answer \citep{fraley_1998_how_many_clust}. Different approaches exist, from fixing $M$ to a pre-specified value \citep{frühwirth_schnatter_2006_fmm,bensmail_1997_modbased_clust}, placing a prior on it \citep{miller_mfm_2018, Nobile_2007_truncpoisson, richardson_1997_Mrandom}, or setting the value deliberately high and letting small concentration parameters in the Dirichlet prior empty redundant components \citep{rousseau_2011_overfitted_mix, malsiner-wall_i2016_sparse_mix}. In the case of a random $M$, the novel method leads to a shrinkage of the number of clusters, eliminating redundant clusters. 

\section{Mixtures and repulsive mixtures}\label{sec:mix_repmix}
In the context of mixture modelling, it is important to highlight the distinction between the number of components and the number of clusters \citep{argiento_infinity_2022}. The number of components  $M$ refers to the number of possible clusters and corresponds to the data generating process, while the number of clusters is the number of allocated components, i.e., components to which at least an observation has been assigned. Such distinction is important, in particular when devising computational strategies. 
The mixture model in \eqref{eq:lik} can be expressed hierarchically through the allocation vector
$\bm c = (c_1, \dots, c_N)$, which indicates the component assignment for each observation:
\begin{align}\label{eq:mix_hierarch}
\bm y_i \mid c_i, \bm \theta_{c_i} &\ind f\left(\bm y_i \mid \bm \theta_{c_i}\right), \quad i=1, \ldots, N \\
\bm \theta_m &\iid P_{\theta}, \quad m=1, \ldots, M \nonumber \\
c_1, \dots, c_N \mid \bm{w} &\iid \mathcal{C}\left(1, \bm{w} \right)  \nonumber \\
\bm{w} &\sim Dir(\alpha) \nonumber \\
M &\sim p_M \nonumber
\end{align}
where $f$ is usually a parametric distribution, $P_{\theta}$ a prior measure for the location parameters $\bm \theta_m$ on $\Theta \in \mathcal{R}^p$, $\mathcal{C}\left(1, \bm{w} \right)$ the categorical distribution with probability vector $\bm w$, $Dir(\alpha)$ is the symmetric Dirichlet distribution with parameter $\alpha > 0$, and $p_M$ a probability mass function on $\{1, 2,...\}$. 
The Bayesian literature suggests various approaches for dealing with the number of components $M$.
The three  main approaches are: (i) setting $M=\infty$ leading to a nonparametric mixture, (ii) keeping $M$ fixed and choose the number of components according to a model choice criterion, (iii) treating $M$ as random and an object of posterior inference. Here we follow the last approach, allowing for a fully Bayesian treatment. 
Possible choices for $p_M$ include an improper uniform prior \citep{richardson_1997_Mrandom}, or a beta-negative-binomial distribution introduced by \citet{fruhwirth_2021_telescope}. Historically, sampling from the posterior distribution of such models requires the implementation of labour-intensive reversible jump Markov chain Monte Carlo (RJMCMC) methods \citep{green_1995_reversible, richardson_1997_Mrandom}, which adaptively alter the dimensionality of the parameter space to account for different values of $M$. Recent developments bridge the trans-dimensional gap using an augmented scheme involving unnormalised weights \citep{argiento_infinity_2022}, or a birth-and-death step \citep{cremaschi_2023_chaos}, which we employ for posterior inference (see  Section~\ref{sec:tiebreaker_mixture}.)

Due to its support on the simplex and conjugacy in a Bayesian framework, the Dirichlet distribution has been the standard choice as a prior for the component weights $\bm w$. It is parametrised by a vector of concentration parameters, with small values shifting the probability mass to the boundaries of the simplex and larger values favouring similar values for the weights by concentrating the probability mass to the centre. \cite{argiento_infinity_2022} extend the class of possible prior distributions by introducing a constructive definition of the mixture weights through the normalisation of a finite point process. Recently, \cite{page_2023_asymmetric_dir} proposed a novel asymmetric Dirichlet prior, which allows the introduction of prior information on the number of clusters in a direct way.
An important aspect of the model specification in~\eqref{eq:mix_hierarch} is the assumption of independent and identically distributed (i.i.d.) component parameters $\bm \theta_m$, motivated also by computational convenience. Such a prior poses no restrictions on the proximity of the cluster centres, which can lead to the creation of multiple overlapping clusters, possibly introducing redundancy. If the primary objective of the analysis is density estimation, such behaviour can be desirable and beneficial. If, on the other hand, the aim is to investigate the underlying cluster structure, the presence of redundant or overlapping clusters can hinder meaningful interpretation.

To mitigate this issue, repulsive mixtures have been proposed with the goal of favouring well-separated component locations. In the Bayesian paradigm, this can be achieved by incorporating a repulsive term into the prior for the component parameters. Motivated by models of interacting particles in statistical mechanics, repulsive mixture models have emerged as a promising solution to the challenges of overfitting and spurious cluster detection. 
Specifically, repulsive mixture models build on the well-established theory of Gibbs point processes (GPPs) and determinantal point processes (DPPs). GPPs, traditionally used in spatial statistics \citep{daley_2003_pointprocess}, describe  the total energy of a configuration of particles. The fundamental idea is that the total energy can be expressed as a sum of energy potentials, accounting for individual particles, pairs, triples, and higher-order interactions. Typically, only the first- and second-order terms are included, the latter of which represent the pairwise potential and thereby the attraction or repulsion between pairs of particles. This idea has been adapted in the context of mixture models by incorporating the pairwise potential into the prior for the component parameters, thereby enabling modelling of repulsion between locations. Notable contributions include the works by \cite{petralia_reomix_2012}, \cite{quinlan_repmix_2021} and \cite{xie_2020_repgmm}, who augment conventional priors by multiplying them with a function used to model the pairwise interaction. A major drawback of the approach based on Gibbs processes is the intractability of their normalising constant for a random number of components, resulting in \cite{petralia_reomix_2012, fuquene_2019_nlp, quinlan_repmix_2021} fixing $M$. \cite{cremaschi_2023_chaos} introduce a novel class of repulsive prior distributions, by setting the joint eigenvalue distributions of random matrices as prior distribution for the locations. They show that for such distributions the Large Deviation Principle holds, ensuring uniqueness and existence for $M\rightarrow\infty$. Moreover, the normalising constant is available in closed form, greatly simplifying computations. Finally, \cite{10.1093/jrsssb/qkaf027}  provide a unified framework for the construction  of random probability measures with interacting atoms, allowing for  both repulsive and attractive behaviours. 

Determinantal Point Processes (DPPs) offer an alternative framework for modelling repulsion in mixture models. Initially studied by \cite{macchi_1975_dpp} to model the distribution of atomic particles at equilibrium, DPPs assume that the distribution of the locations is proportional to the determinant of a positive semi-definite matrix. \cite{xu_repddp_2016}, \cite{bianchini_ddpremix_2020} and \cite{beraha_2021_repmix_mcmc} apply DPP priors in mixture models with a random number of components, at the price of elaborate and complex Markov chain Monte Carlo algorithms to sample from the intractable posterior.

The main contribution of this paper is the introduction of a mixture model that incorporates repulsion in both the component weights and locations, along with an investigation of its properties in terms of the number and size of clusters. The Large Deviation Principle (proved in the Supplement) ensures the existence and uniqueness of the infinite dimensional object from which the Selberg Dirichlet distribution is derived for $M$ finite. 

\section{The SIP mixture}
\label{sec:tiebreaker_mixture}
\subsection{The SIP prior}
As distribution for the weights of a mixture, we specify the \textit{Selberg Dirichlet} \citep{pham_selberg_2009} which is an extension of the well-known Dirichlet distribution. 
\begin{definition}[Selberg Dirichlet]
The M-dimensional Selberg Dirichlet distribution describes a random vector $\bm{w} = (w_1, \ldots, w_M)$ defined on the simplex, where $0 \leq w_m \leq 1$ for all $m$ and $\sum_{m = 1}^M w_m = 1$. The density function of this distribution is:
\begin{align}
\label{eq:sdir_prior}
&SDir(\bm{w}, \alpha, \gamma, M) =  
\frac{1}{D(\alpha, \gamma, M)} 
\left( \prod_{m = 1}^{M}  w_m^{\alpha - 1} \right)  |\triangle \bm{w}|^{2\gamma}\\
&D(\alpha, \gamma, M) = \frac{\Gamma(\alpha)}{\Gamma(M\alpha + \gamma (M - 1)(M-2))}
\prod_{j = 1}^{M - 1} \frac{\Gamma(\alpha + (j - 1)\gamma)\Gamma(1 + j \gamma) }{\Gamma(1 + \gamma)} \nonumber
\end{align}
where the parameters satisfy $\gamma \geq 0$ and $\alpha > 0$. The term $\triangle \bm{w} = \prod_{1 \leq i < j \leq M - 1} |w_i - w_j|$ represents the product of the absolute pairwise differences among $M - 1$ components of $\bm{w}$. 
\end{definition}

Similar to the repulsive priors in \cite{petralia_reomix_2012, quinlan_repmix_2021,xie_2020_repgmm} and \cite{cremaschi_2023_chaos}, the Selberg Dirichlet distribution contains a term corresponding to a standard distribution (the Dirichlet) and  a repulsive term.

Figure~\ref{fig:sdir_prior} shows the Selberg Dirichlet distribution for $M = 3$ as a ternary plot for varying values of $\gamma$ and $\alpha$. When $\gamma = 0$, the Selberg Dirichlet reduces to the Dirichlet, represented in the first column of  Figure~\ref{fig:sdir_prior}. For the remaining plots, the probability becomes 0 for combinations where the first two values are close, visually indicated by the dashed gray lines on the main axes. Similarly to the Dirichlet distribution, a higher $\alpha$ leads to more prior mass in the centre of the simplex, while for smaller values of $\alpha$ the probability mass concentrates near the boundary. Conversely, higher values of the repulsion parameter favour realisations near the boundary of the simplex,  highlighting the repulsive property of the distribution. 

For a vector $\bm{w} = (w_1, \ldots, w_M)$ distributed according to the M-dimensional Selberg Dirichlet distribution and defined on the simplex, where $0 \leq w_m \leq 1$ for all $m$ and $\sum_{m = 1}^M w_m = 1$, the following proposition holds.
\begin{proposition}
Expectations, higher order moments, as well as the variance of $w_j$,  are given by
\begin{align*}
    \mathbb{E} \left\{ \prod_{i=1}^{M} w_i \right\} &= \frac{D(\alpha+1,\gamma,M)}{D(\alpha,\gamma,M)}\\
    \mathbb{E} \left\{ w_j \right\} &= \frac{\alpha}{\alpha M + (M - 1)(M - 2)\gamma} \\
    \mathbb{E} \left\{ \prod_{i=1}^{M} w_i^{k} \right\} &= \frac{D(\alpha+k, \gamma, M)}{D(\alpha, \gamma, M)} \\
 \mathbb{E} \left\{ w_j^k \right\} &= \frac{\Gamma(\alpha+k)
 \Gamma(\alpha M + (M-1)(M-2)\gamma)
 }{\Gamma(\alpha)\Gamma(\alpha M + k + (M-1)(M-2)\gamma)}\\     
    \mathbb{E} \left\{ w_j^2 \right\} &= \frac{\alpha(\alpha+1)}{(\alpha M + 1 + (M-1)(M-2)\gamma)(\alpha M + (M-1)(M-2)\gamma)}\\
    \mathbb{V} \left\{ w_j \right\} &= \frac{\alpha}{k} \left( \frac{1 - \frac{\alpha}{\eta}}{\eta + 1} \right)
\end{align*}
where $ \eta = \alpha M + (M - 1)(M - 2)\gamma$.
\end{proposition}
\renewcommand{\qedsymbol}{}
\begin{proof}
    See supplementary material.
\end{proof}
\renewcommand{\qedsymbol}{$\square$}
The variance of each component of the Selberg Dirichlet random variable is decreasing in $\gamma$ (proof provided in the supplementary material), which implies that the variance is always higher for a standard Dirichlet (i.e., when $\gamma = 0$). Beyond standard moments, we also consider a measure of internal dispersion \citep{pham_selberg_2009}:
\[
\theta_\tau(\bm{w}) = \left| \prod_{1 \leq i < j \leq M-1} (w_i - w_j) \right|^\tau
\]
whose expectation is
\[
\mathbb{E} \left\{ \theta_\tau(\bm w) \right\} = \frac{D(\alpha, \gamma + \frac{\tau}{2}, M)}{D(\alpha, \gamma, M)}
\]
for $\tau \geq 0$. Analogous to the determinant of the covariance matrix, this measure summarises the overall dispersion of a random vector into a single scalar value. Figure~\ref{fig:exp_intern_disper} displays the expected internal dispersion across different values of $\alpha$, $M$, $\gamma$, and $\tau$. As the dispersion parameter $\gamma$ increases, the internal dispersion rises, whereas higher values of $\alpha$, $M$, or $\tau$ lead to a reduction in the measure. 

\begin{figure}[tp]
\centering
     \includegraphics[width=.22\linewidth]{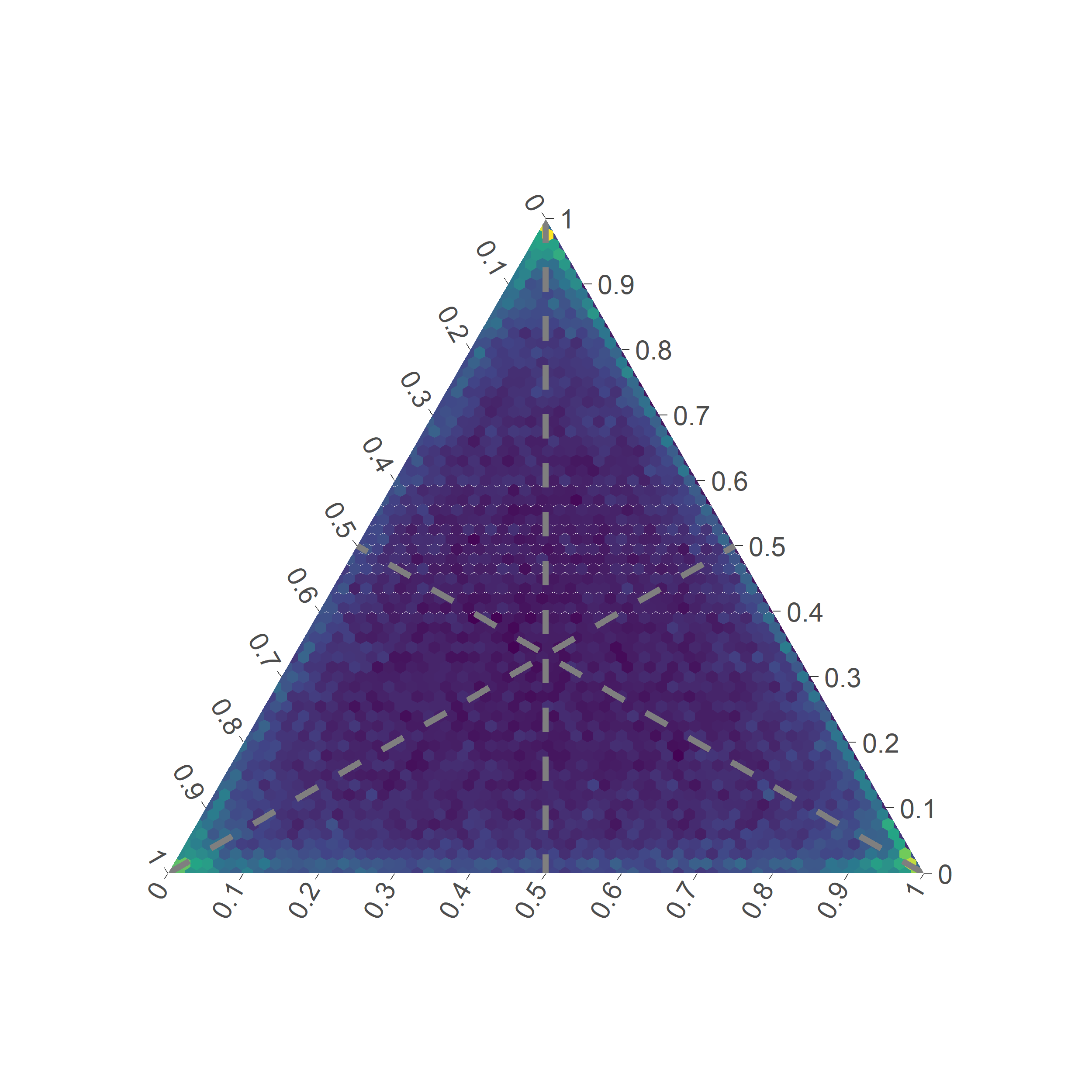}
    \includegraphics[width=.22\linewidth]{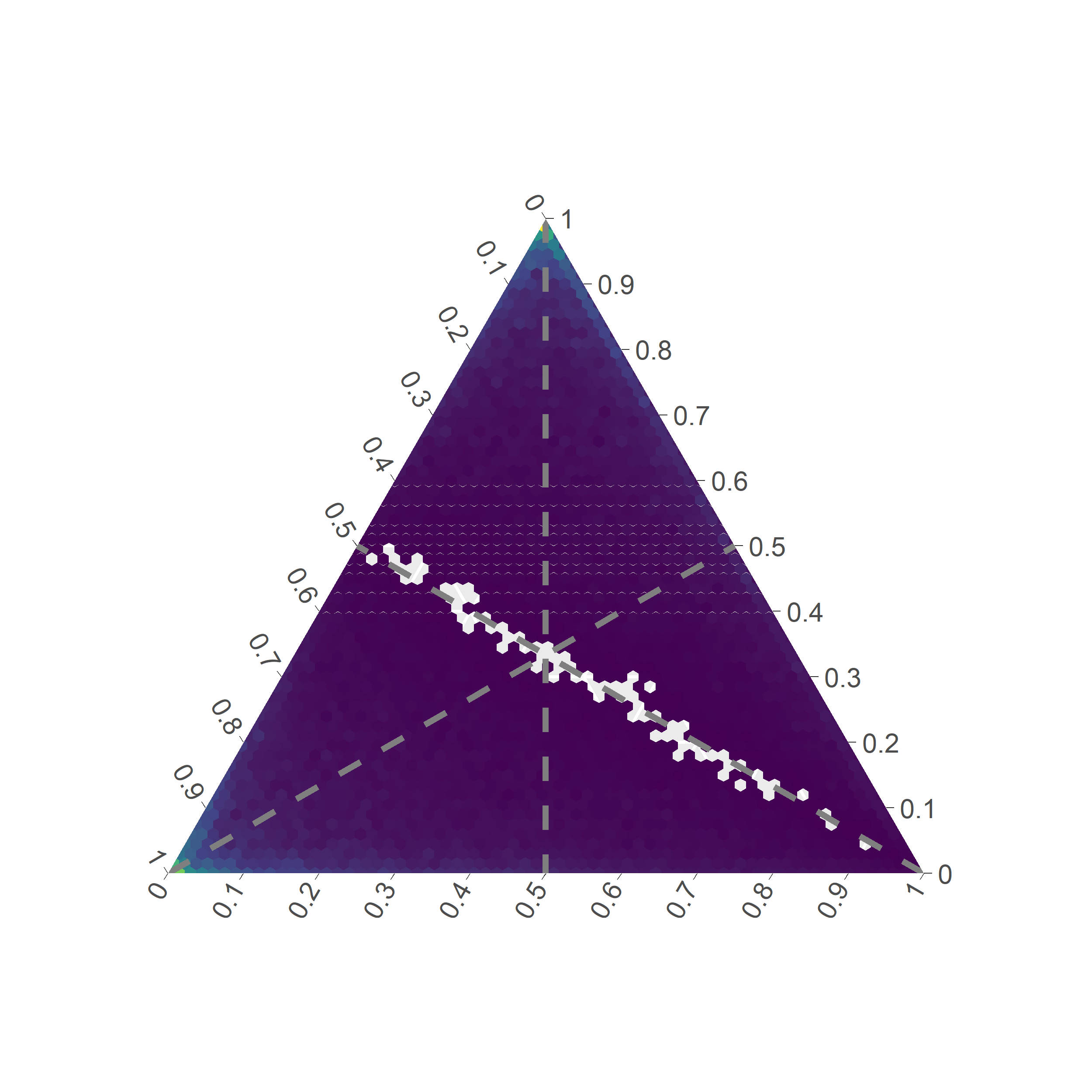}
     \includegraphics[width=.22\linewidth]{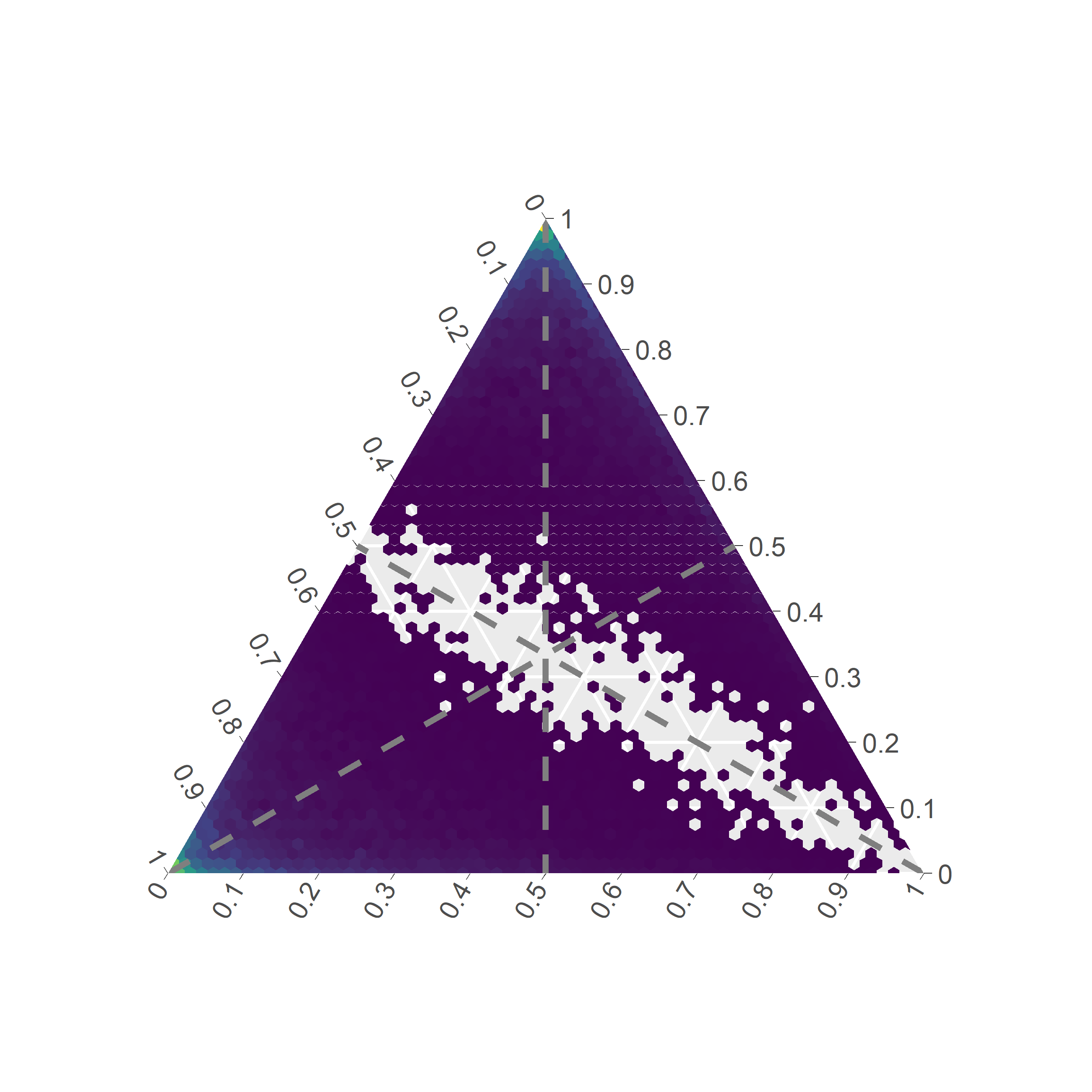}
    \includegraphics[width=.22\linewidth]{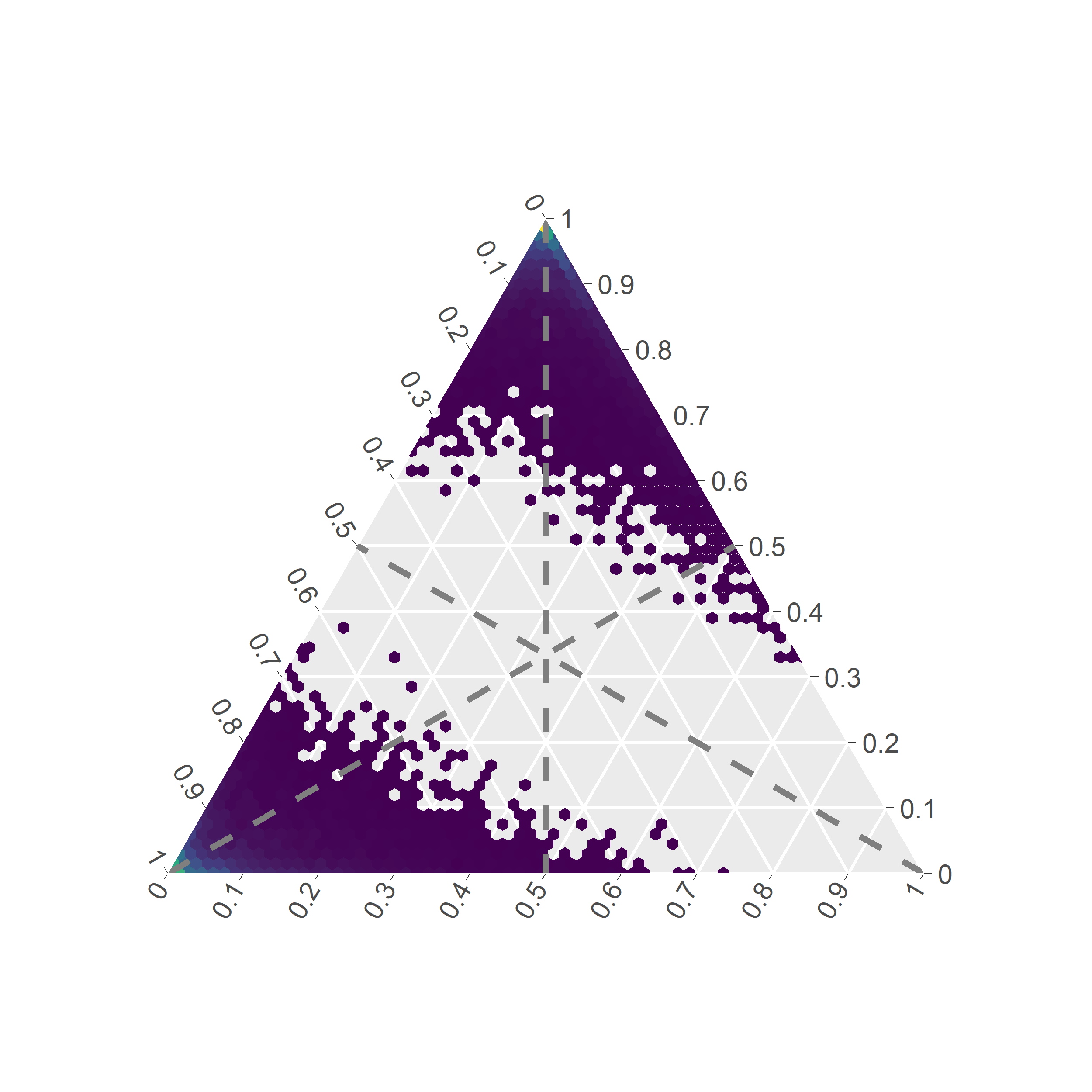}\\
    \includegraphics[width=.22\linewidth]{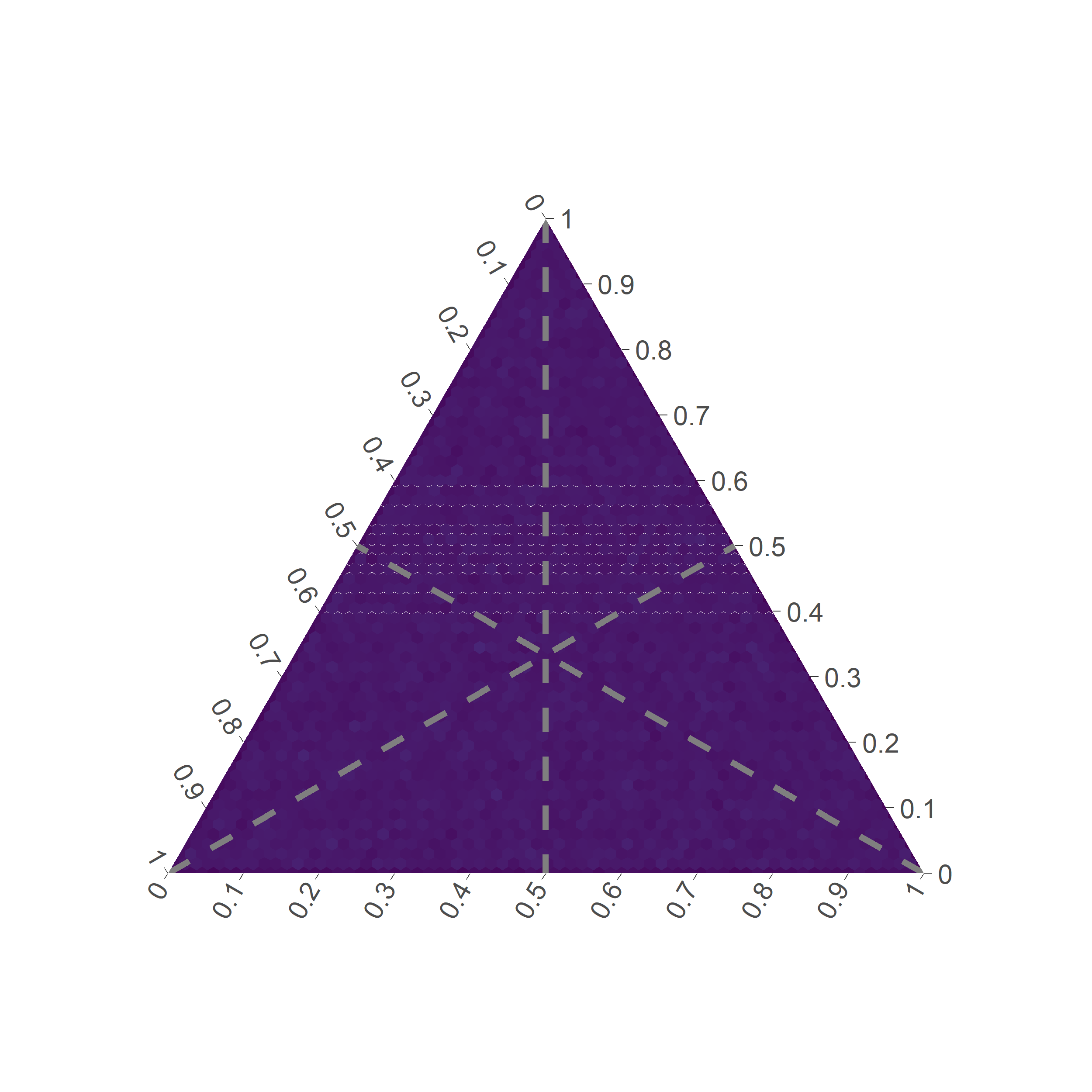}
    \includegraphics[width=.22\linewidth]{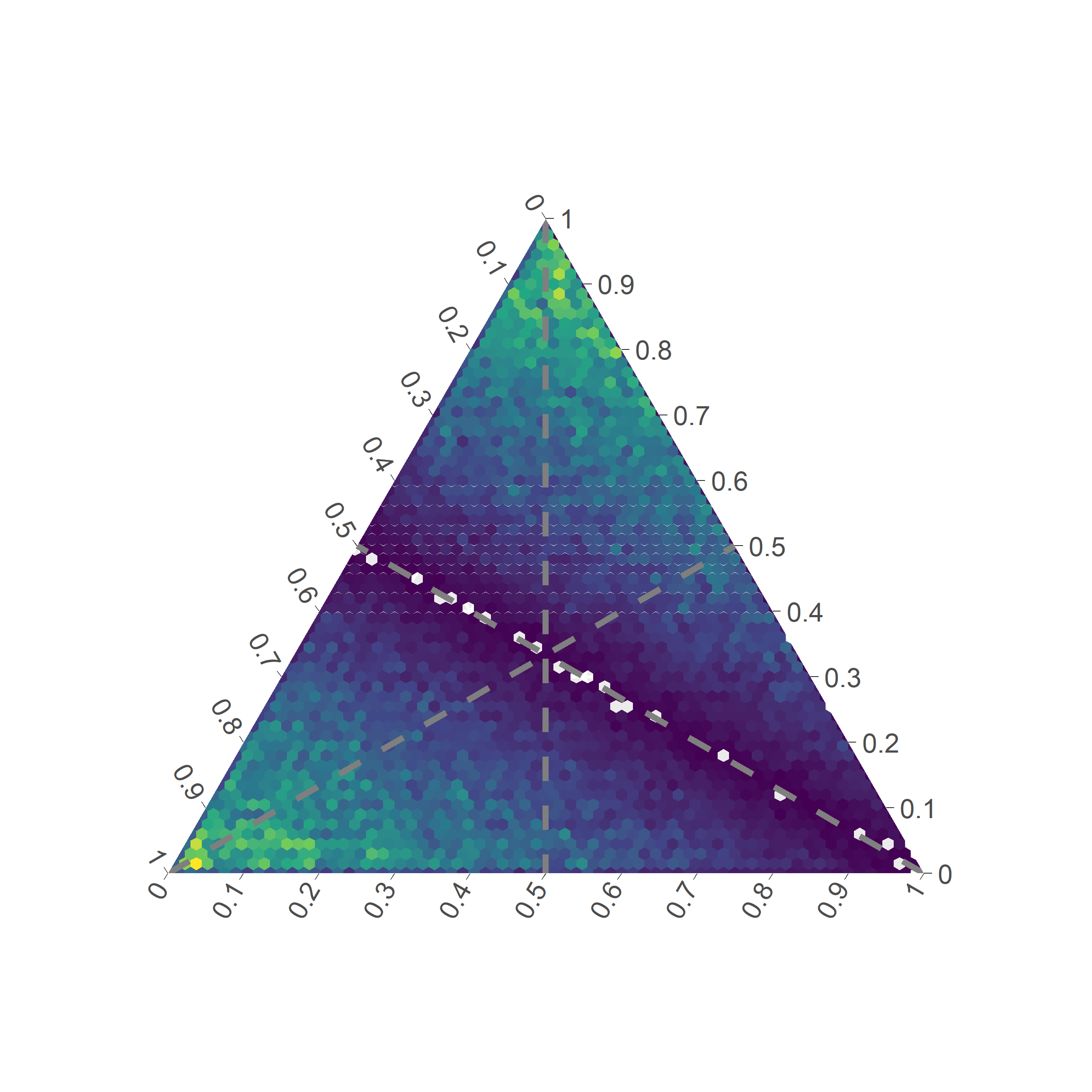}
    \includegraphics[width=.22\linewidth]{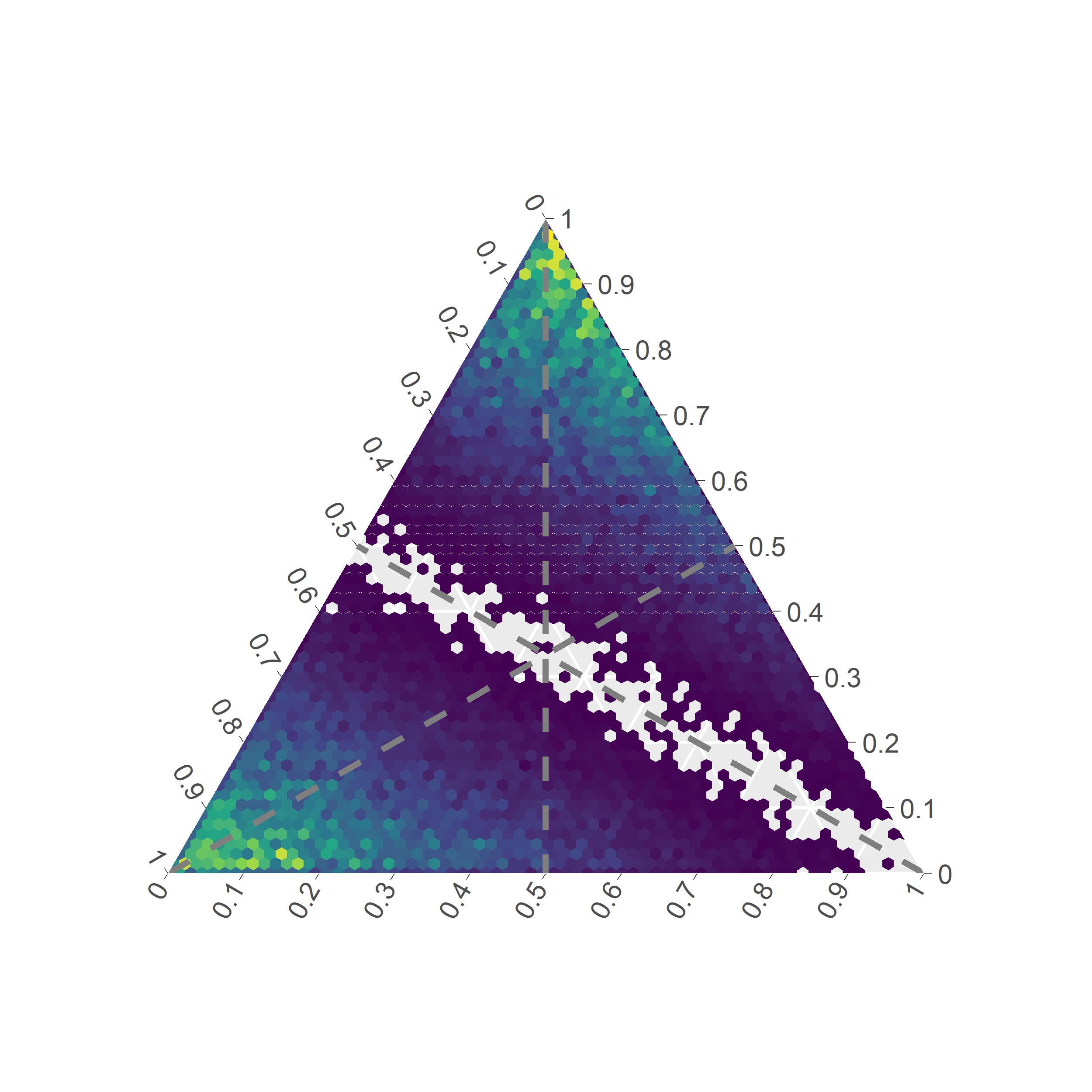}
    \includegraphics[width=.22\linewidth]{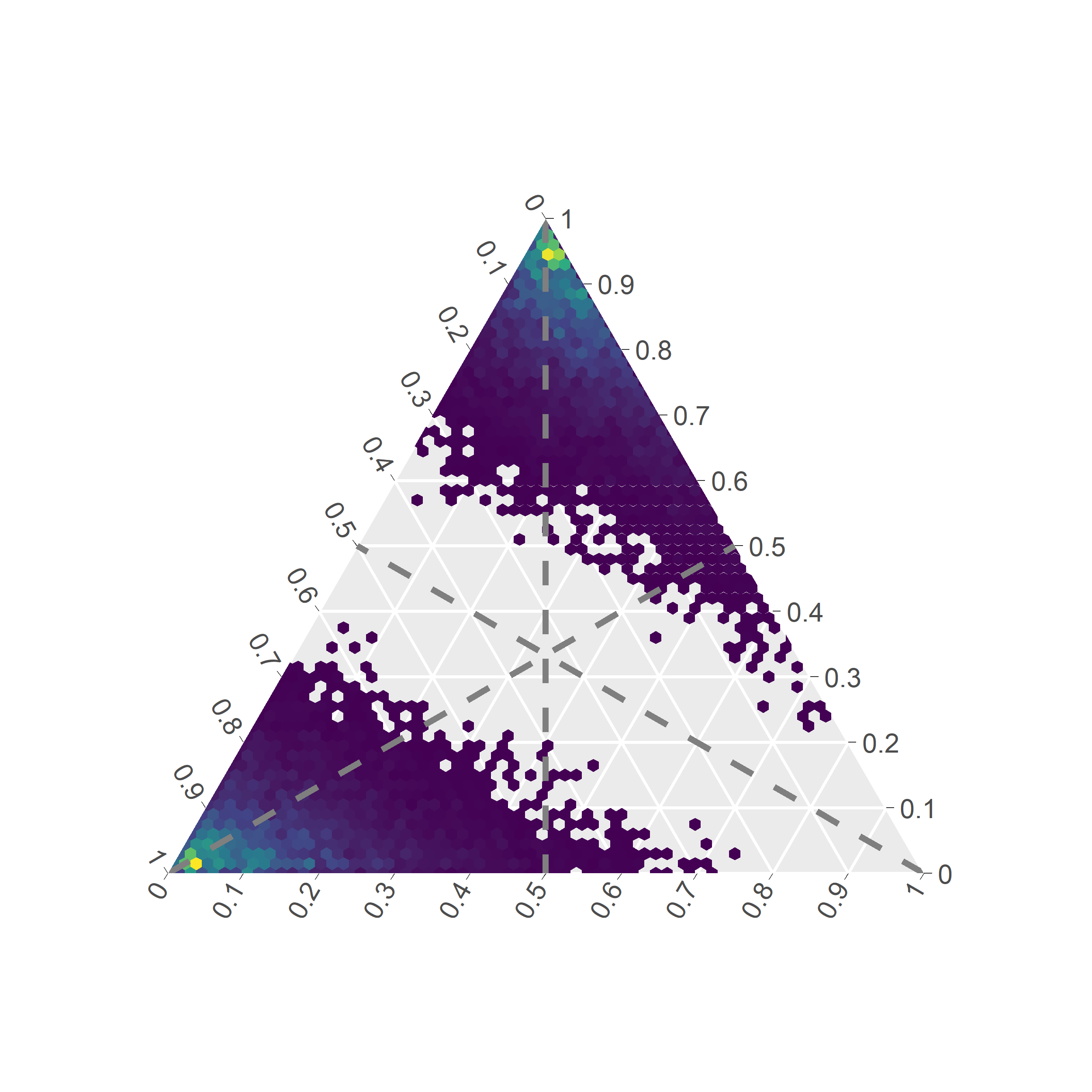}
\caption{Simulation-based densities of the Selberg Dirichlet distribution for varying values of the repulsion parameter $\gamma$ and concentration parameter $\alpha$. From left to right: $\gamma = (0, 0.5, 1, 3)$. From top to bottom: $\alpha = (0.5, 1)$. For illustration purposes, the colours are scaled separately for each plot.}
\alttext{A series of eight ternary plots arranged on a 2 by 4 grid presenting the density of the Selberg Dirichlet distribution. Rows show concentration parameter alpha = 0.5 (top) and alpha = 1 (bottom). Columns show repulsion parameter gamma increasing from 0 (Dirichlet) to 0.5, 1 and 3. As gamma increases, the probability density concentrates near the vertices and edges, avoiding the corresponding axis, visually demonstrating the repulsive property.}
\label{fig:sdir_prior}
\end{figure}

\begin{figure}[tp]
\centering
 \includegraphics[width=0.9\textwidth]{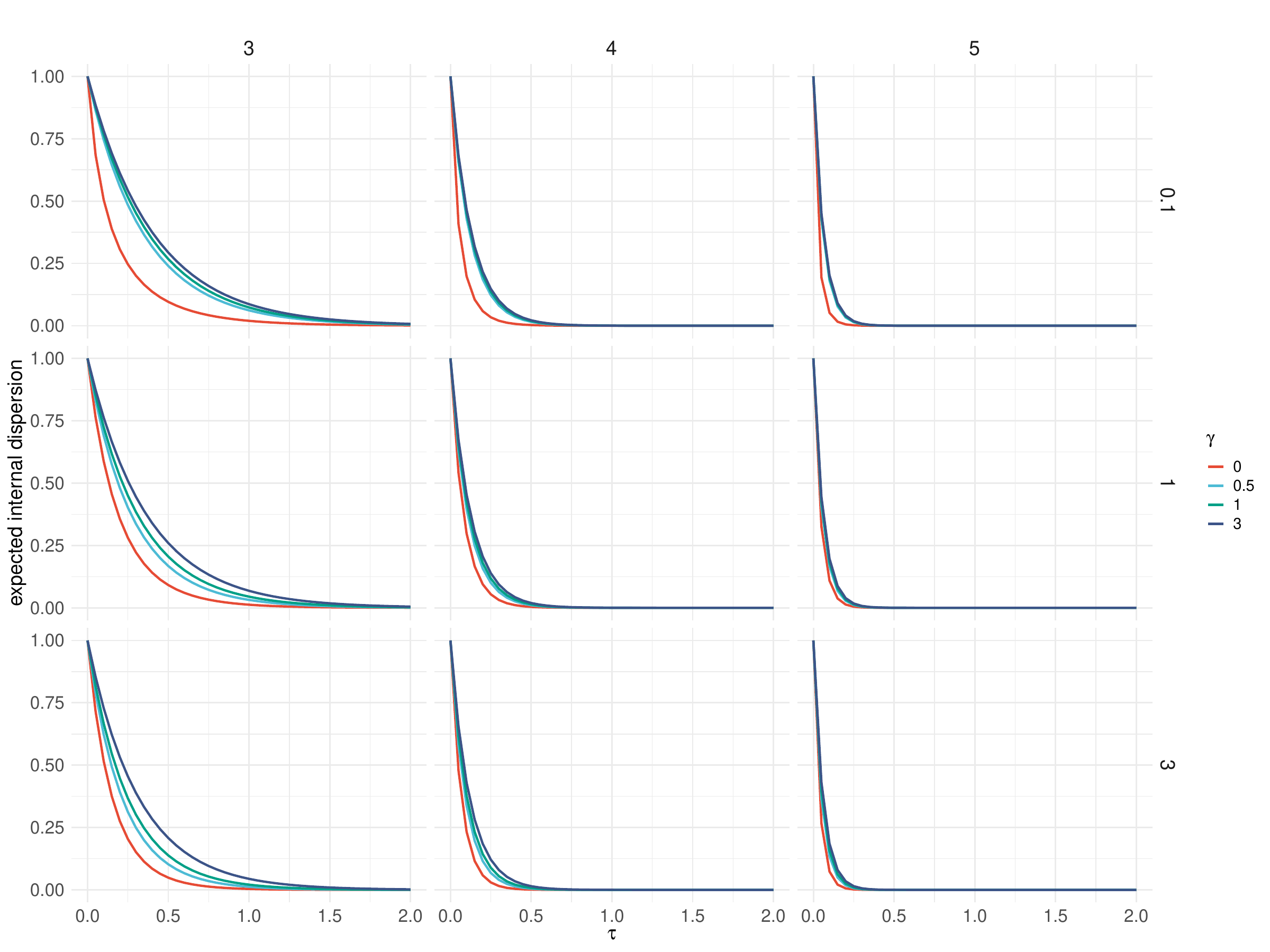}
\caption{Expected internal dispersion for varying parameters: columns correspond to \( M = 3, 4, 5 \) (left to right), and rows to \( \alpha = 0.1, 1, 3 \) (top to bottom).}
\alttext{Nine line graphs on a 3 by 3 grid. Each graph shows the expected internal dispersion (vertical axis) against the parameter tau (horizontal axis) for different values of the repulsion parameter gamma (colored lines). Dispersion increases with gamma but decreases as the number of components M, concentration parameter alpha, or tau increases.}
\label{fig:exp_intern_disper}
\end{figure}

\subsection{A Large Deviation Principle}

\label{sec:LDP}
In this section, we present a Large Deviation Principle (LDP) for the generalised Selberg Dirichlet. LDPs are useful tools in statistical physics, where they can be used to deduce properties such as the existence and uniqueness of the Gibbs measures \citep{georgii2011gibbs}, and in random matrix theory, where they can be used to show that the empirical spectral distribution of a random matrix converges to some deterministic measure \citep{agz2009rmt}. In a more elementary context, an LDP essentially says that the probability that a histogram of i.i.d. data looks like a different distribution than the one it was generated from decays exponentially. In our case, we use the LDP to justify allowing an arbitrary number of components, which is related to its application in showing the existence an uniqueness of Gibbs measures.

Given a sequence of probability measures \(\{\mu_n\}\) on a topological space \(T\), a lower semicontinuous function \(I: T \to [0,\infty]\), and a sequence of positive real numbers \(\{r_n\}\), we say that \(\{\mu_n\}\) satisfies a large deviation principle (LDP) with rate function \(I\) and speed \(r_n\) if the following two conditions hold:
\begin{enumerate}
    \item For every open set \(G \subset T\),
    \[
    \liminf_{n \to \infty} \frac{1}{r_n} \log \mu_n(G) \geq -\inf_{x \in G} I(x)
    \]
    \item For every closed set \(F \subset T\),
    \[
    \limsup_{n \to \infty} \frac{1}{r_n} \log \mu_n(F) \leq -\inf_{x \in F} I(x)
    \]
\end{enumerate}
If, in addition, the level sets \(\{x \in T : I(x) \leq c\}\) are compact for all \(c < \infty\), then the rate function \(I\) is called a \emph{good} rate function.

If $\mathcal{A}$ is a base for the topology, then the LDP is equivalent to the conditions \citep{dembo2009large}
\begin{align}
-I(x)  & = \inf \left\{ \limsup_{n \to \infty} \frac{1}{r_n} \log \mu_n(G) : G \in \mathcal{A}, x \in G \right\} \\ &  = \inf \left\{ \liminf_{n \to \infty} \frac{1}{r_n} \log \mu_n(G) : G \in \mathcal{A}, x \in G \right\}
\end{align}
Intuitively, we can think of these as conditions on the tails of the distribution. Now, to obtain our theorem, we must show that
\begin{equation}\label{LDP1}
-I(\mu)  \geq \inf \left\{ \limsup_{n \to \infty} \frac{1}{r_n} \log \mu_n(G) : G \right\}
\end{equation}
and
\begin{equation}\label{LDP2}
-I(\mu)  \leq \inf \left\{ \liminf_{n \to \infty} \frac{1}{r_n} \log \mu_n(G) : G \right\}
\end{equation}
where $G$ runs over neighbourhoods of $\mu$.

In fact, the conditions above are equivalent only to a weak LDP, in which case another property of the sequence of measures, exponential tightness, is required to establish the full LDP \citep{RAS2015LDPgibbs}. In our case, becasue the underlying space is compact, we do not need this additional property for the full LDP. 
\begin{theorem}
\label{thm:ldp}
Suppose $\bm w = (w_1,\dots,w_M) \sim SDir$ and let $\mu_{E_M}$ be the associated empirical measure. Suppose that $\lim_{M \to \infty} \frac{M}{\alpha_k -1} \to a$.
Then, the limit
\[ \lim_{M \to \infty} \frac{1}{M^2} \log B_M := B < \infty \]
exists and $\mu_{E_M}$ satisfies an LDP in the scale $M^{-2}$ with rate function
\[  I(\mu):= \iint F(x,y) d\mu(x)d\mu(y) = -2a^2\gamma\iint \log|x-y| d\mu(x)d\mu(y) - 2a\int \log x d\mu(x) + B \]
for $\mu \in \mathcal{M}[0,1]$. Moreover, there exists a unique $\mu_0$ such that $I(\mu_0) = 0$.
\end{theorem}
\renewcommand{\qedsymbol}{}
\begin{proof}
    See supplementary material.
\end{proof}
\renewcommand{\qedsymbol}{$\square$}

\subsection[Implied prior on the number of clusters]{Implied prior on the number of clusters $M^a$}
\label{sec:implied_prior}
We now discuss the prior distribution of the number of clusters induced by the proposed prior, represented by the number of allocated components $M^a$ \citep{argiento_infinity_2022, frühwirth_schnatter_infinity_2019, quinlan_repmix_2021, nobile_mixpost_2004, miller_mfm_2018}. 

We can obtain a visual representation of the prior on $M_a$ induced by the Selberg Dirichlet by simulating from the SIP mixture with a fixed number of components. For \( M \in \{3, \dots, 8\} \) and a fixed concentration parameter \( \alpha_0 = 1 \), we vary the repulsion parameter \( \gamma \in \{0, 1, 3\} \). In the simulation, we fix the number of observations to $n=100$. Figure~\ref{fig:prior_pred_M_bar} clearly shows that increasing repulsion leads to a substantially lower number of expected clusters compared to the standard Dirichlet prior case (\( \gamma = 0 \)).


\begin{figure}[t!]
\centering
\includegraphics[width=.9\linewidth]{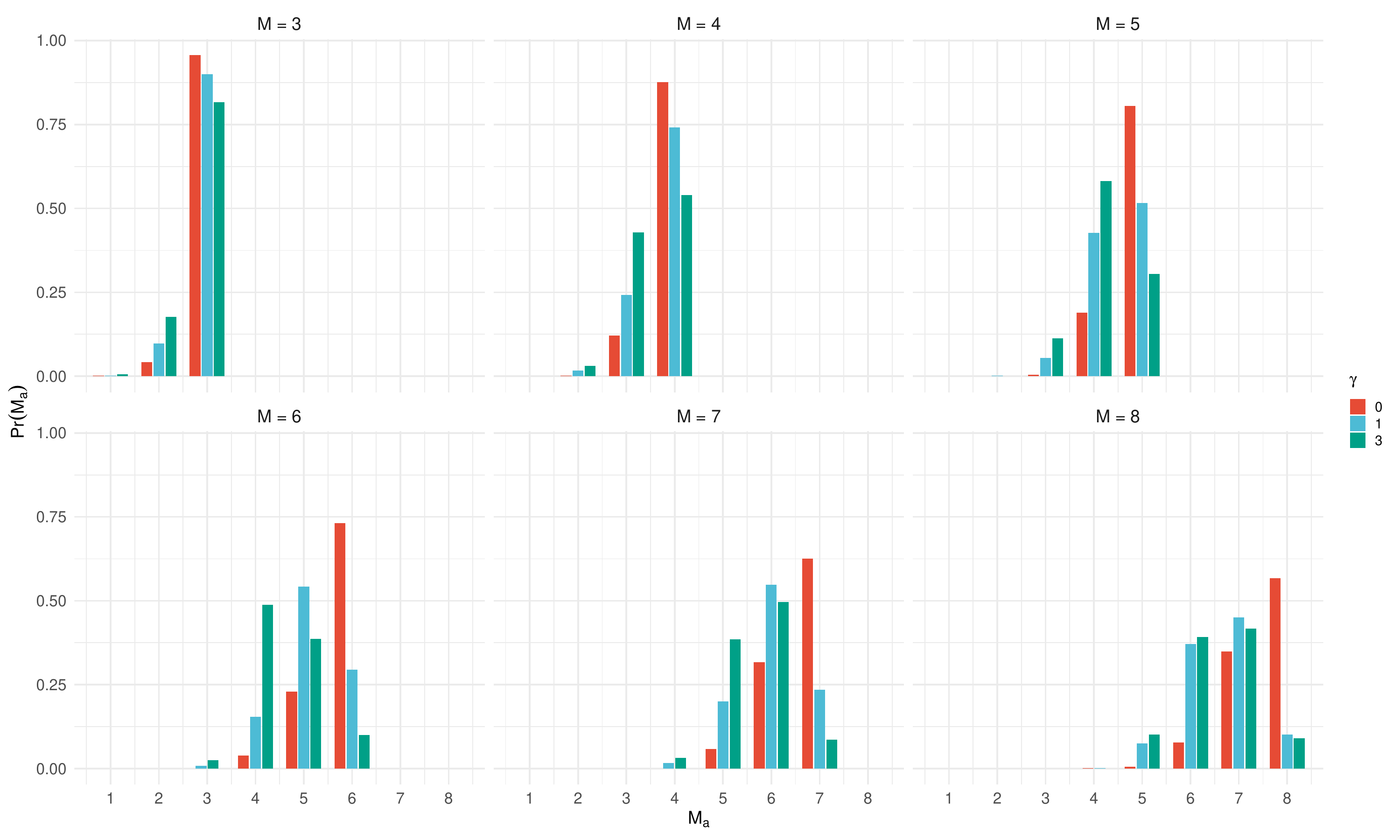}
\caption{Simulation-based, implied prior on $M_a$ for $\alpha_0 = 1$ and varying values of $M$ and $\gamma$.}
\alttext{Barplots on a 2 by 3 grid showing the distribution of the number of allocated clusters. The six panels correspond to different numbers of components (M = 3 to 8). Within each panel, the color of the bars varies with the repulsion parameter (gamma = 0, 1, 3). For a fixed M, the bars shift left (fewer clusters) as gamma increases, illustrating the sparsity-inducing effect of the SIP prior.}
\label{fig:prior_pred_M_bar}
\end{figure}

For the Selberg Dirichlet distribution, shrinkage on the  number of clusters can be introduced via both the repulsion parameters $\gamma$ as well as $\alpha_0$. \cite{rousseau_2011_overfitted_mix}, show that if $\alpha_0 < \frac{d}{2}$, where $d$ represents the dimension of the cluster-specific parameters, the posterior expectation of mixture weights associated with redundant clusters asymptotically approaches zero, guaranteeing the automatic emptying of superfluous components. Moreover, \cite{malsiner-wall_i2016_sparse_mix} show that, in practice, a very small value of $\alpha_0$ is necessary to accurately identify the correct number of clusters, in the case of overfitted mixture with shrinkage priors on the component locations. Our results suggest that the SIP prior can further enhance the sparsity even for moderate to large values of $\alpha_0$, for a general class of priors on mixture components.
\subsection{Posterior inference}\label{sec:post_inf}
We consider a mixture model with random number of components, which is tractable thanks to the fact that the normalising constant of the Selberg Dirichlet distribution is available in closed form. For the remainder of the manuscript, we will focus on $D$-dimensional Gaussian component densities, each defined by a mean parameter $\bm \mu_m = (\mu_{m, 1},\dots, \mu_{m, D})$ and covariance matrix $\bm \Sigma_m$, for $m = 1, \dots, M$. We denote with $\bm \mu$ and $\bm \Sigma$ all mean parameters and covariance matrices. We opt for inverse-Wishart priors $\inwish{\bm V_0, \nu_0}$ for $\bm \Sigma$ and apply a product of $D$ independent Gaussian ensemble (GE) distributions as priors for $\bm \mu$ to induce repulsion among the component locations \citep{cremaschi_2023_chaos, forrester2010log}. The Gaussian ensemble distribution has density function defined by
\begin{align*}
& GE(\bm x,  M, \zeta) = \frac{1}{\mathcal{G}(M, \zeta)} \prod_{m=1}^M e^{-\frac{\zeta}{2} x_m^2}  \left| \triangle \bm x \right|^\zeta\\
& \mathcal{G}(M, \zeta)=\zeta^{-\frac{M}{2}-\zeta M(M-1) / 4}(2 \pi)^{\frac{M}{2}} \prod_{j=0}^{M-1} \frac{\Gamma\left(1+(j+1) \frac{\zeta}{2}\right)}{\Gamma\left(1+\frac{\zeta}{2}\right)}
\end{align*}
where $\zeta > 0$ and $\triangle \bm{x} = \prod_{1 \leq i < j \leq M} x_i - x_j$. Using this distribution as a joint prior for the component locations, we can favour centres that are more distant to each other, depending on the magnitude of the repulsion parameter $\zeta$. The tractable form of the normalising constant makes this distribution particularly convenient  in mixture models, especially in scenarios involving a prior on $M$, as it is essential in the trans-dimensional birth-and-death step when creating or deleting a component. For the prior on the number of components, we use a Poisson distribution shifted to have support on the positive integers \citep{Nobile_2007_truncpoisson} with rate parameter $\lambda$, denoted as $Poi_1(\lambda)$. An advantage of our model is that it allows for posterior inference on the repulsion parameters \( \gamma \) and \( \zeta \). To enable this, we specify suitable hyperpriors and consider two alternatives. In the first approach, we assign independent Gamma priors to both \( \gamma \) and \( \zeta \). In the second approach, we place a Gamma prior on \( \gamma \) alone and define \( \zeta \) through a fixed ratio \( \rho = \zeta / \gamma \). In Section~\ref{sec:simstudy} we additionally present a data-driven strategy for setting the repulsion parameters.

The hierarchical representation now becomes
\begin{align*}\label{eq:mix_hierarch}
\bm y_i \mid c_i, \bm \mu_{c_i}, \Sigma_{c_i} 
&\ind \mathcal{N}( \bm \mu_{c_i}, \Sigma_{c_i}), 
\quad i = 1, \ldots, N \\
 \mu_{1,d}, \dots,  \mu_{M,d} &\sim GE( \zeta), 
\quad d = 1, \ldots, D \nonumber \\
\bm \Sigma_m &\iid \mathcal{IW}(\bm V_0, \nu_0),
\quad m = 1, \ldots, M \nonumber \\
c_1, \dots, c_N \mid \bm w 
&\iid \mathcal{C}(1, \bm w) \nonumber \\
\bm w &\sim \text{SDir}(\alpha, \gamma, M) \nonumber \\
M &\sim \text{Poi}_1(\lambda) \nonumber \\
\gamma &\sim \Gammad{\alpha_{\gamma, 0}, \beta_{\gamma, 0}}  \nonumber \\
\zeta &\sim \Gammad{\alpha_{\zeta, 0}, \beta_{\zeta, 0}}  \nonumber
\end{align*}
To update the number of components, we employ a birth-and-death step that consists of either adding (birth) or removing (death) a non-allocated component with probabilities $q$ and $1-q$, respectively \citep{geyer_birthdeath_1994}. We develop a tailored MCMC algorithm for posterior inference, whose details can be found in the Supplement.

\section{Simulation study}
\label{sec:simstudy}
After exploring the induced prior on the number of clusters in Section~\ref{sec:implied_prior}, we shift our focus to examining its posterior distribution, in particular the effect of the repulsive parameters on the posterior distribution of $M_a$.
We generate $n = 300$ observations from a mixture of $5$ bivariate Gaussian components, with weights given by $\bm w = (0.2,0.2,0.2,0.3,0.1)$ and cluster means defined as $\bm \mu_{1:5, 1}=(-3,-3,3,3,-1)$ and $\bm \mu_{1:5, 2}=(-2.5,3,-3,3,0)$. The covariance matrices of the first 4 components have entries equal to $3$ on their diagonal and to $1$ on their off-diagonal entries. The covariance matrix of the last component is a diagonal matrix with entries equal to $0.25$. 

Figure~\ref{fig:simstudy_2_data_tsm}(a) shows the simulated dataset, with data associated to each component enveloped in a convex hull of different colour, as well as the corresponding true similarity matrix in Figure~\ref{fig:simstudy_2_data_tsm}(b), ordered using hierarchical clustering with complete linkage for visualisation purposes. The clusters were purposefully designed to overlap, and the fifth cluster added to mimic a redundant cluster to create a problem where the ground truth is extremely hard to recover.  

\begin{figure}[t!]
\centering
\subfloat[]{\includegraphics[width=.45\linewidth]{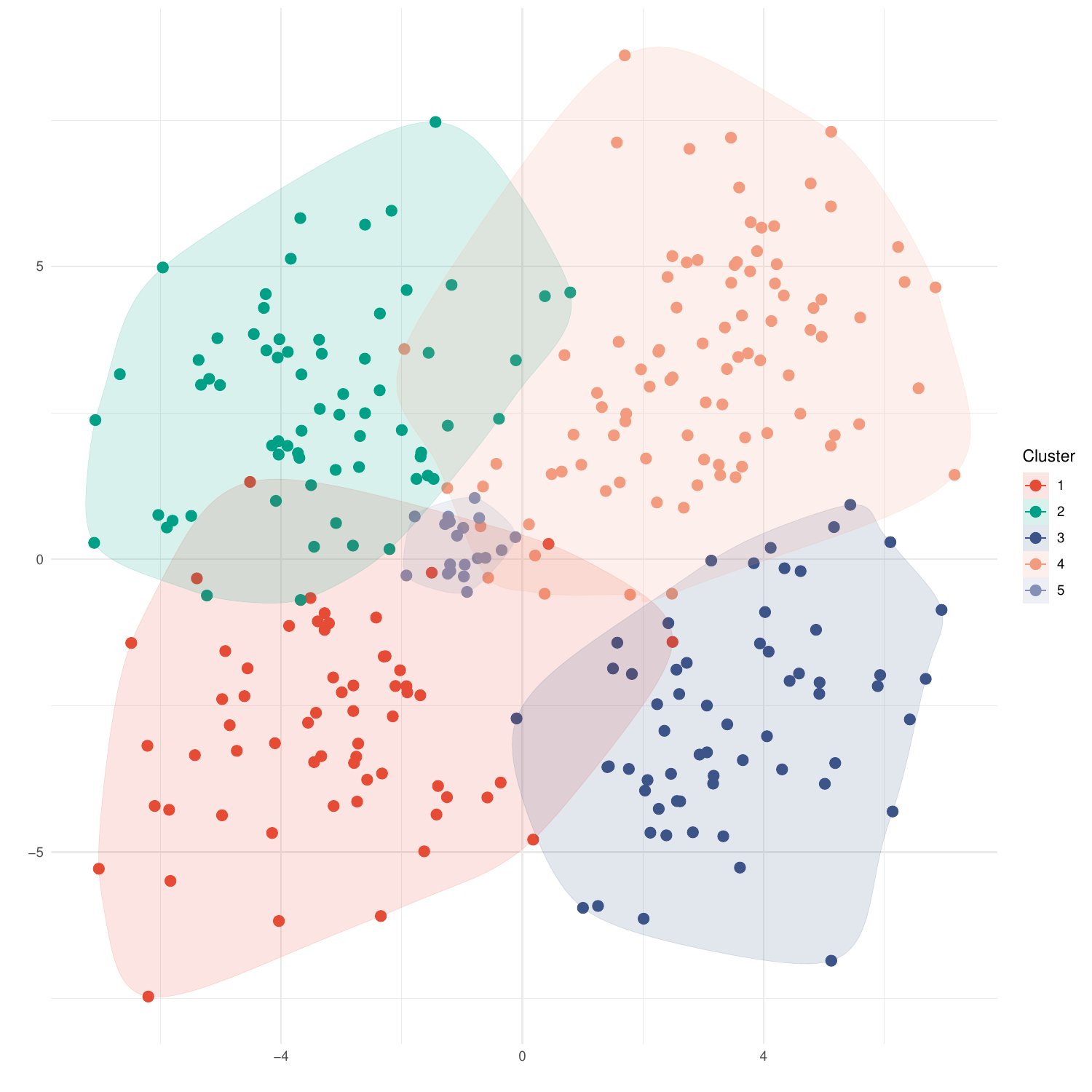}}
\subfloat[]{\includegraphics[width=.45\linewidth]{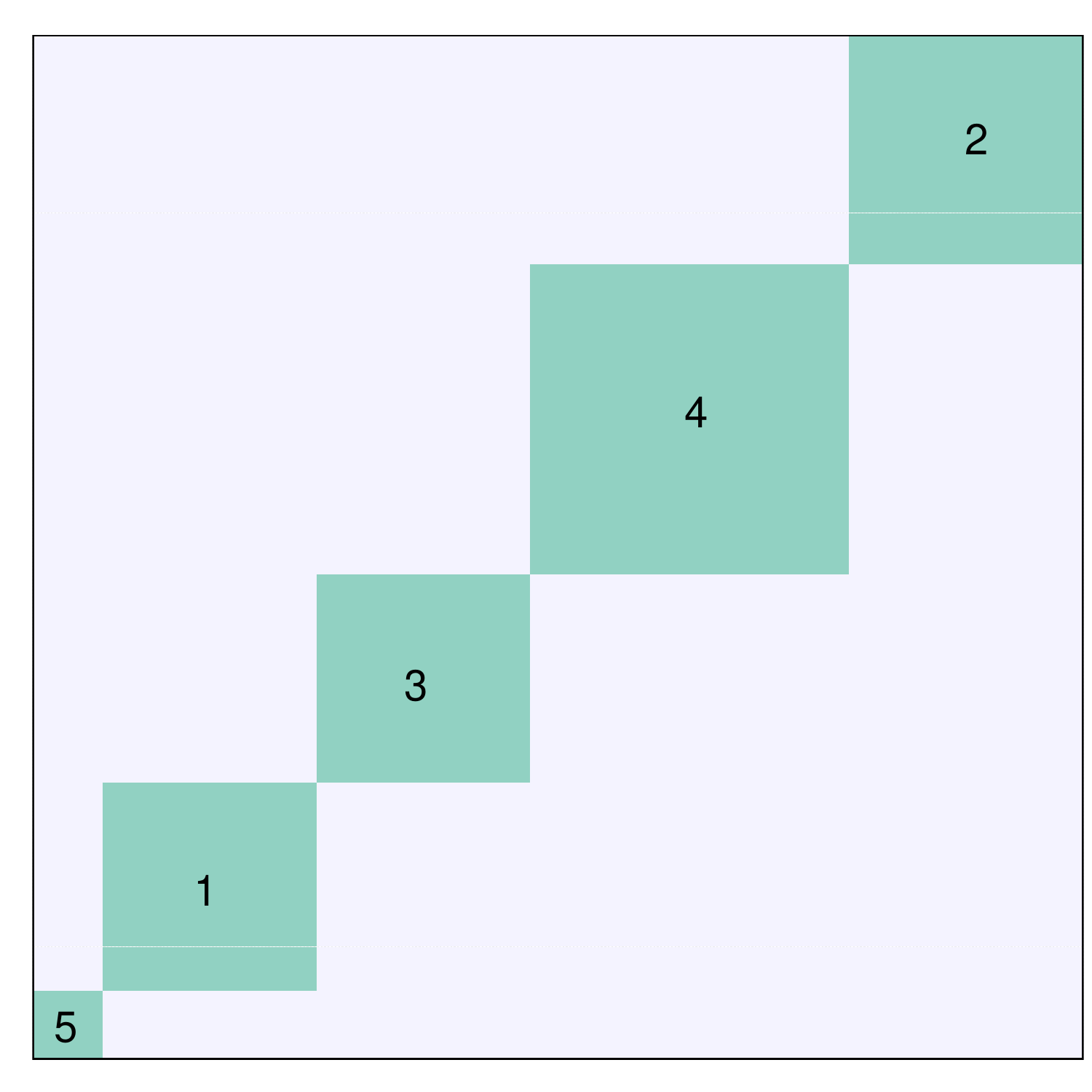}}
\caption{Simulated dataset consisting of $5$ different clusters and the corresponding similarity matrix.}
\alttext{(a) Scatterplot of 300 points from 5 partially overlapping Gaussian clusters, each enclosed in a coloured convex hull. (b) A 300 by 300 similarity matrix where a green square indicates two points originate from the same true cluster, and white indicates different clusters. The matrix is ordered to group points from the same cluster together.}
\label{fig:simstudy_2_data_tsm}
\end{figure}

We run the MCMC algorithm introduced in Section \ref{sec:post_inf}, and sample 5000 posterior draws, after discarding 5000 as burn-in period and retaining every $10$-th draw. We employ two independent Gaussian ensemble distributions as priors for $\bm \mu_{1:M, 1}$ and $\bm \mu_{1:M, 2}$ and $\Sigma_1, \dots, \Sigma_5 \sim \inwish{\mathbb{I}_2, 2}$, where $\mathbb{I}_2$ denotes the identity matrix of dimension $2$. We set $\lambda = 3$, so that $M$ is centred around the target value, $\alpha_0 = 1$, and the probability of a birth move $q$ is set to $0.5$.
%
%
%
Figure~\ref{fig:Ma_bar_psm_fixed} shows the posterior number of clusters and the posterior similarity matrix for different combinations of hyperparameters and hyperpriors in the model.

To guide the choice of $\zeta$, we first run the $k$-means algorithm on the dataset, setting the number of clusters to the true value $K = 5$. We then compute the average distance between cluster centres along each dimension. Next, we simulate values from the Gaussian ensemble prior across a grid of $\zeta$ values and calculate the corresponding average distances. The value of $\zeta$ that produces the closest match to the cluster centre distances obtained from $k$-means is found to be $\zeta =0.1$, shown in the first column. Regarding the hyperprior on $\gamma$ shown in the last row, we set $\gamma \sim \Gammad{3,2}$, so that a-priori mean and variance are equal to one.

We can observe that an increase in $\gamma$, corresponding to lower rows in the figure, leads to both a decrease in the posterior mean and variance of $M_a$. On the other hand, an increase in $\zeta$, from left to right in the figure, leads to a higher number of clusters. When $\gamma \sim \Gammad{3,2}$, its posterior concentrates around $0.2$, and the posterior number of clusters lies between the outcomes obtained under fixed values $\gamma = 0$ and $\gamma = 0.25$.

The variability in the posterior distribution of the partition can be appreciated by looking at the posterior similarity matrices. All plots in the third row, corresponding to $\gamma = 1$, display mostly two or three clusters with almost no variability. However, while the first two plots show clearly separated clusters, the last exhibits significant variability in the co-clustering probability. This suggests that when both $\gamma$ and $\zeta$ are large, their joint effect increases variability in cluster assignments. 



\begin{figure}[t!]
\centering
\subfloat[]{\includegraphics[width=.45\linewidth]{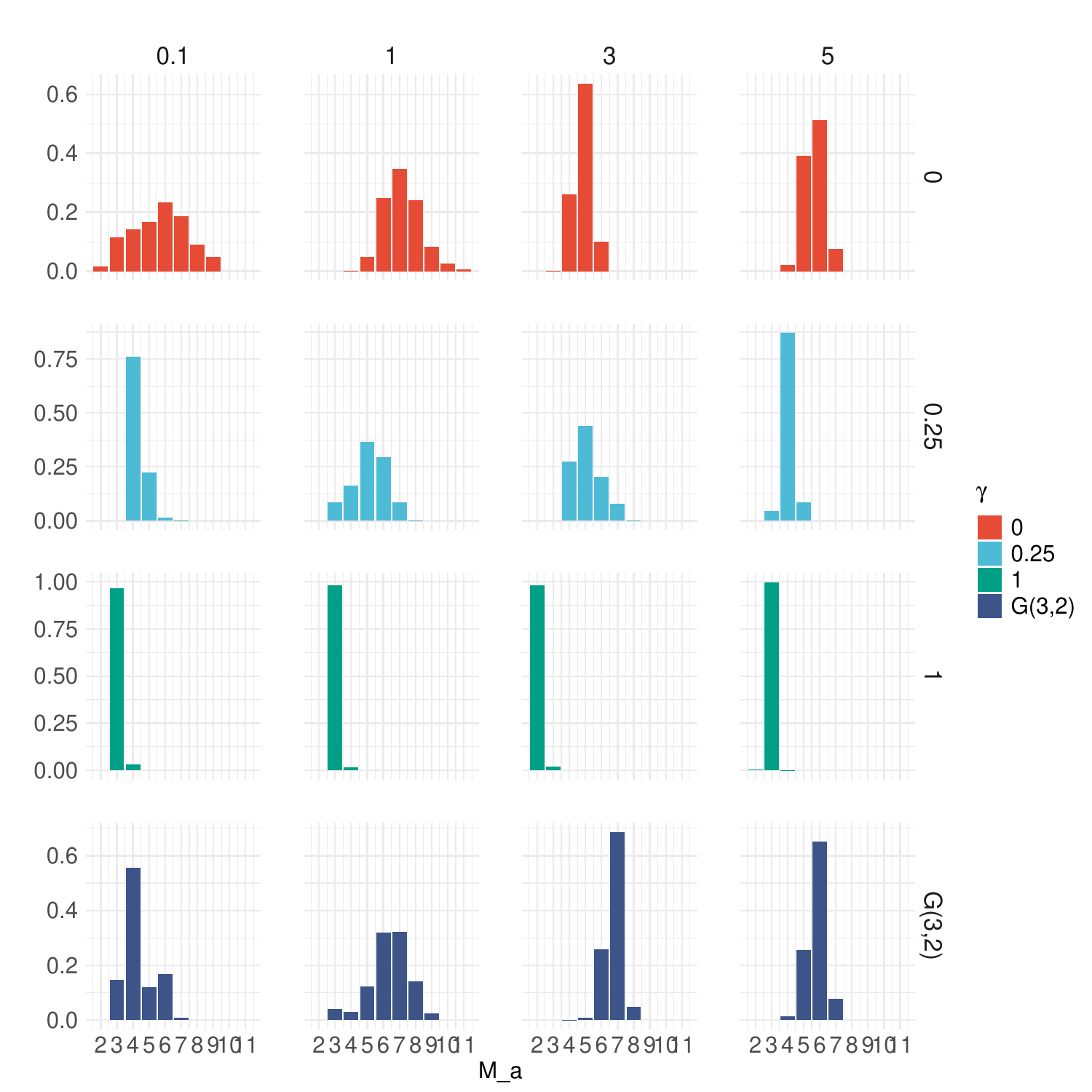}}
\subfloat[]{\includegraphics[width=.45\linewidth]{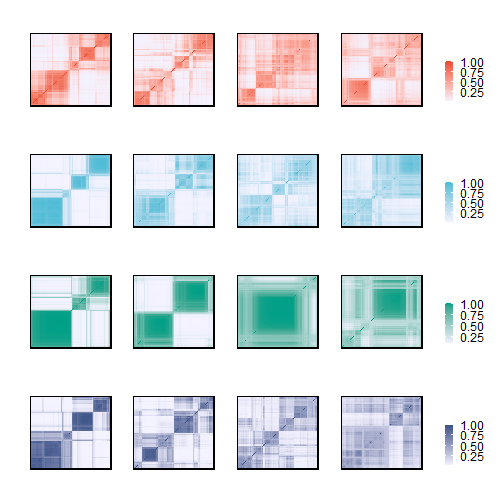}}
\caption{(a) Posterior distribution of $M_a$ and (b) posterior similarity matrices for varying combinations of $\gamma$ and $\zeta$.}
\alttext{(a) Barplots on a 4 by 4 grid showing the posterior distribution of the number of clusters $M_a$ for different fixed values and a hyperprior of the repulsion parameter gamma (rows) and fixed values of zeta (columns). Higher gamma reduces cluster count, higher zeta increases it. (b) 4 by 4 grid of posterior similarity matrices. Darker squares indicate higher probability that two observations belong to the same cluster.}
\label{fig:Ma_bar_psm_fixed}
\end{figure}


\begin{figure}[t]
\centering
 \includegraphics[width=\linewidth]
 {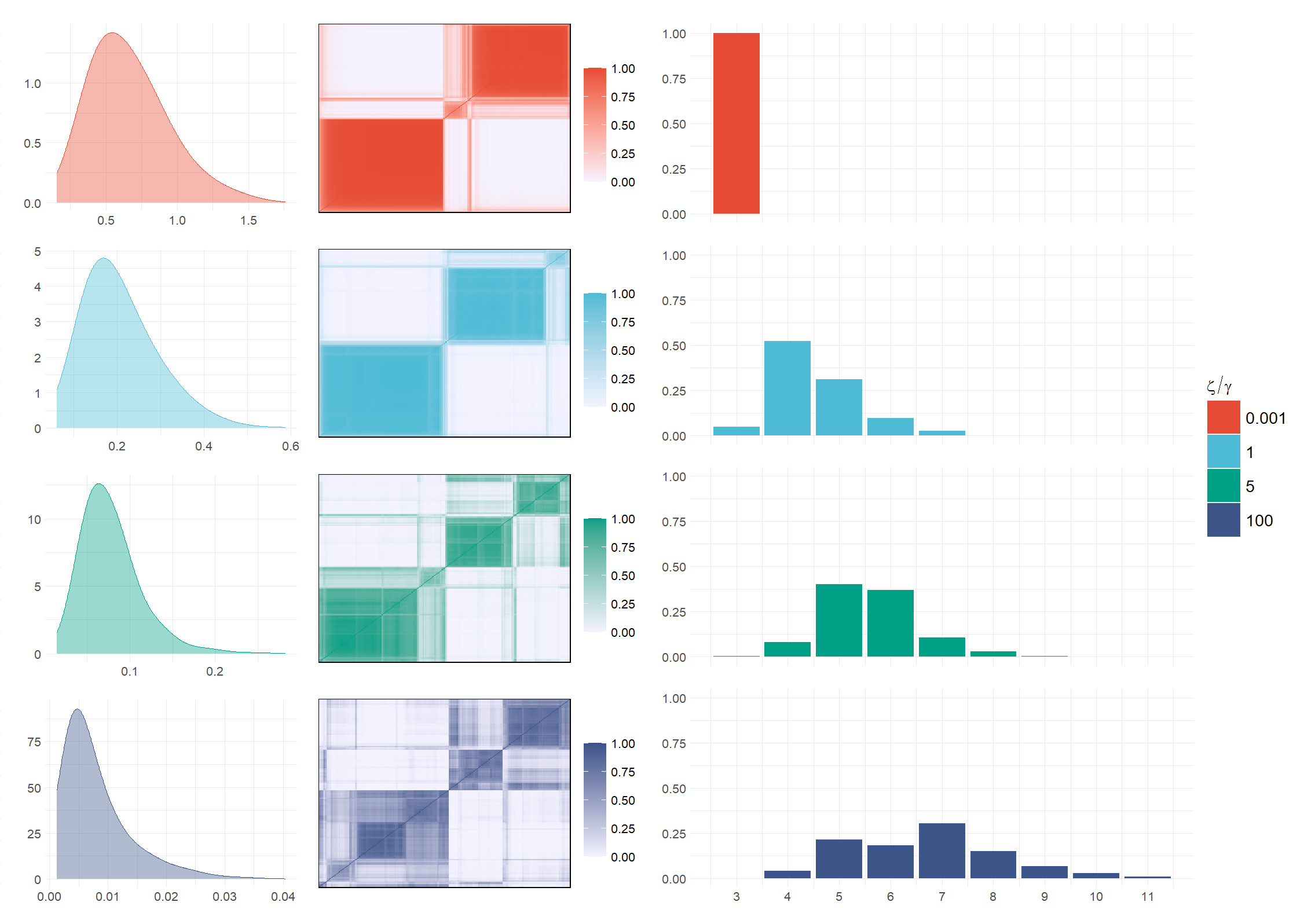}
\caption{From left to right: posterior density plots of $\gamma$, posterior similarity matrices and posterior number of clusters $M_a$. The ratio of $\zeta$ and $\gamma$ is fixed at 0.001, 1, 5 and 100 (from top to bottom).}
\alttext{Figure with four rows, each corresponding to a different fixed ratio rho = zeta/gamma and three columns. Within each row, the first column shows the density plots of the posterior for gamma, the second column shows a posterior similarity matrix, and the third column shows the posterior distribution of the number of clusters. As rho increases from 0.001 to 100, the posterior distribution of the number of clusters shifts toward higher values, the similarity matrices display more variability, and the estimated gamma decreases.}
\label{fig:ratio_posterior_4col}
\end{figure}

In Figure~\ref{fig:ratio_posterior_4col} we present results corresponding to the SIP mixture for the same hyperprior on $\gamma$ and with $\zeta$ defined via the ratio $\rho = \frac{\zeta}{\gamma}$, where $\rho \in \{0.001, 1, 5, 100\}$. On the leftmost column, we show the kernel density estimates of the posterior of $\gamma$, in the second column the posterior similarity matrix and in the last column the posterior distribution of the number of clusters $M_a$.

An increase in $\rho$, warranting an increase in $\zeta$, increases the number of clusters, which aligns with our previous findings. This can be observed in both the posterior of $M_a$ and the column depicting the posterior similarity matrices. The plots showing the posterior of $\gamma$, placed in the first two columns, indicate that $\gamma$ is inversely related to $\zeta$. This is especially visible for $\rho = 100$, where $\gamma$ is pushed towards the value $0.01$. 

The simulated dataset mimics a real-world scenario where the definition of clusters should be guided by the research question and expert judgment. We demonstrate that the SIP mixture model can recover different shapes and structures by appropriately adjusting the repulsive parameters. For problems requiring a more non-informative prior stance, modelling the ratio of repulsive parameters offers a simplified approach, allowing modellers to express their prior beliefs through a single parameter, $\rho$.
\section{Application to data on children's BMI and eating behaviour}
\label{sec:empirical}
We apply the SIP repulsive mixture model to a dataset including the standardised body mass index (Z-BMI) and eating behaviour of $n = 537$ children from the Singaporean cohort GUSTO (Growing Up in Singapore Towards healthy Outcomes) \citep[][]{soh_2013_gusto}. GUSTO is a highly phenotyped prospective cohort started in 2009 and still ongoing, collecting a plethora of information on more than a thousand mother-child dyads. The Z-BMI is defined as the standardised ratio of a person's weight in kilograms and their height in meters squared, and the eating behaviour was assessed using the children's eating behaviour questionnaire \citep{wardle_2001_cebq_origin}. The questionnaire comprises 35 items relating to one of 8 subscales, which can be classified as either describing the food-approach (food responsiveness, enjoyment of food, emotional overeating, and desire to drink) or food-avoidance (slowness in eating, satiety responsiveness, food fussiness, and emotional undereating) of children. Instead of modelling the questionnaire answers directly, which are of ordinal nature, we use the partial credit model \citep{masters_1982_pcm} to recover the latent traits, or person parameters, quantifying their approach to food and beverage intake. Confirming results found in existing research, which relates food-approach traits to a propensity to eat more, we find a positive correlation of $0.24$ between its latent variable and the Z-BMI \citep{fogel_2017_bmi}.


We fit the SIP mixture with random number of components to the Z-BMI and eating behaviour latent traits data. We fix $\lambda = 2$ and $\alpha_0 = 0.5$. The prior distributions employed for the component locations and covariance matrices are the same as for the simulation study in~\ref{sec:simstudy}, namely two independent Gaussian ensemble distributions and an $\inwish{\mathbb{I}_2, 2}$ distribution. Regarding the repulsive parameters $\gamma$ and $\zeta$, we present results for three different combinations below. The birth and death probabilities are equal to $0.5$. We analyse the posterior results after discarding 5000 samples as burn-in period and retaining every 10-th value for a final sample of 5000. We compute posterior estimates of the random partition by minimising the Binder loss function, a popular choice in Bayesian clustering analysis \citep{binder_1978_binder,wade_2023_bayclust_analysis}.

Figure~\ref{fig:gusto}(a) shows the results of a model configuration with moderate repulsion on the component locations and negligible repulsion on the weights, corresponding to $\gamma = 0.1$ and $\zeta = 1$. The model identifies four clusters, with one (light blue) overlapping another (red), highlighting a common limitation of standard mixture models and motivating the use of repulsive priors. By increasing the repulsion on the weights ($\gamma = 2$) and on the locations ($\zeta = 3$), we can decrease the estimated number of clusters to $K = 2$. By adjusting the repulsion on both parameters to $\gamma = 2$ and $\zeta = 1$, we are able to ameliorate this shortcoming and find three interpretable clusters. Focusing on the result depicted in Figure~\ref{fig:gusto}(c), the first cluster (red) contains children with average Z-BMI and food-approach traits parameters, indicating a healthy approach to food which is reflected in their Z-BMI. The second cluster (light blue) is characterised by children with a Z-BMI above 1 and a slightly larger dispersion across their food-approach traits. The third cluster (green) shows children that are both associated with lower Z-BMI and food-approach. This group involves children with leaner body types, potentially linked to a diminished enjoyment of food and lower instances of emotional overeating.
Significant differences in maternal BMI before pregnancy across the three clusters were detected by a one-way ANOVA (p-value = $1.01 \times 10^{-7}$), validating existing findings that relate maternal pre-pregnancy BMI and offspring growth \citep{michael2023longitudinal}. In contrast, a multinomial Cochran-Armitage trend test \citep{Szabo_2019_multiom_cochrane} found no significant differences in maternal education levels between the three estimated clusters (p-value = 0.6932). Corresponding boxplots and barplots are provided in the supplementary material.

\begin{figure}[tp]
  \centering
  \subfloat[$\gamma = 0.1$, $\zeta = 0.1$]{\includegraphics[width=.33\linewidth]{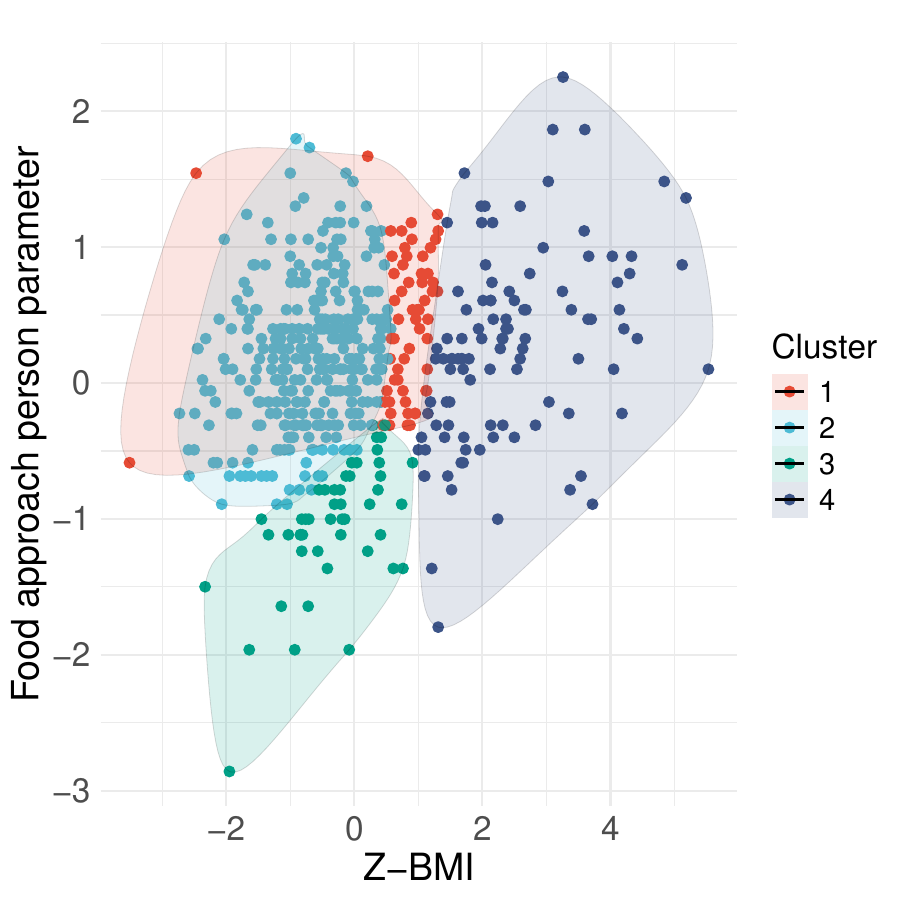}}
  \centering
  \subfloat[$\gamma = 2$, $\zeta = 3$]{\includegraphics[width=.33\linewidth]{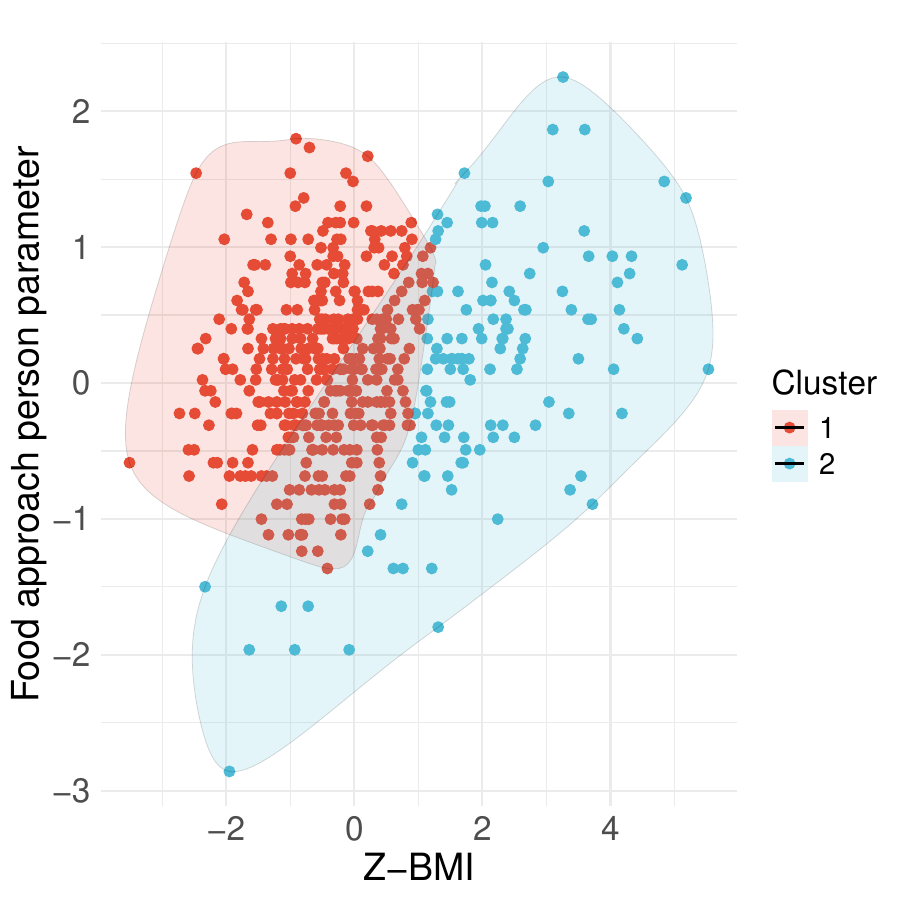}}
  \centering
  \subfloat[$\gamma = 1$, $\zeta = 3$]{\includegraphics[width=.33\linewidth]{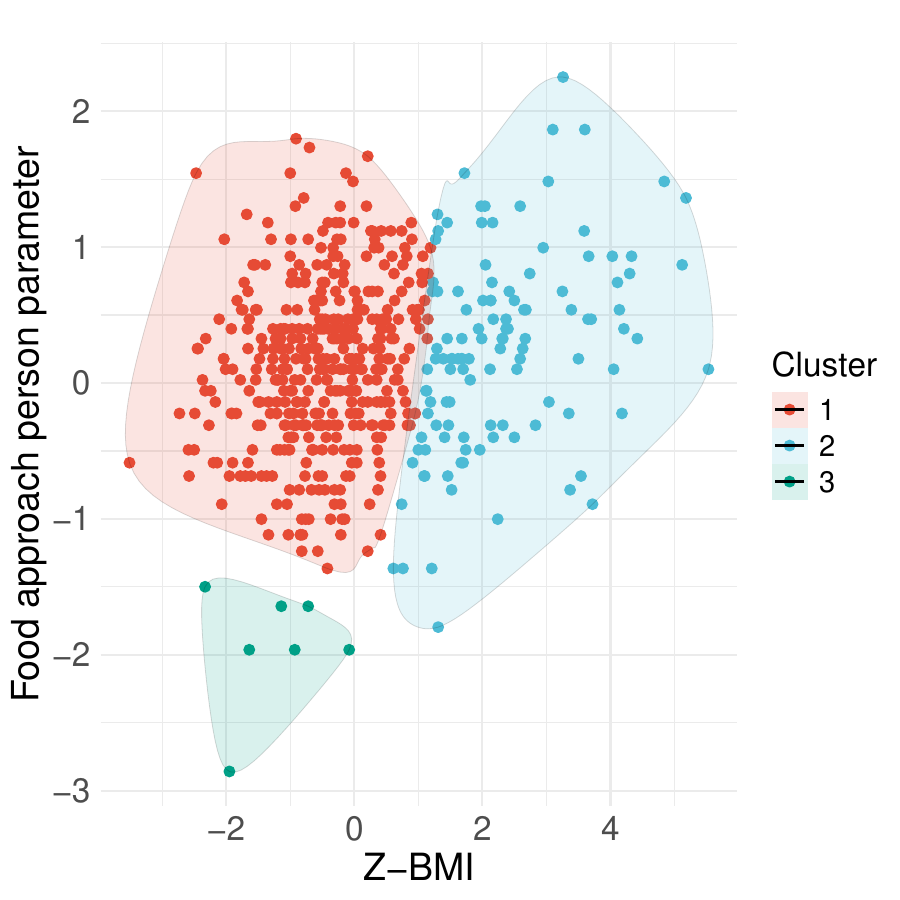}}
\caption{Point estimates of the posterior clustering based on the Binder loss function for the SIP mixture.}
\alttext{Three scatterplots showing children's standardized body mass index (Z-BMI) versus a latent food-approach trait. Each point is colored by its assigned cluster from the SIP mixture. Panel (a): four clusters with substantial overlap. Panel (b): two clusters with slight overlap. Panel (c): three distinct and well-separated clusters — one average cluster (red), one cluster with high Z-BMI (blue), and one cluster with low Z-BMI and low food-approach (green).}
\label{fig:gusto}

\end{figure}
\section{Conclusion}
\label{sec:conclusion}
This paper introduces the SIP prior, a novel method that incorporates repulsion directly into the component weights of a mixture model. The approach is built upon an extension of the Dirichlet distribution known as the Selberg Dirichlet, which introduces pairwise repulsive potentials between particles, a concept inspired by statistical physics that underpins the effectiveness of repulsive mixtures. By promoting dissimilarity among weights, increased repulsion under the Selberg Dirichlet prior naturally encourages solutions with fewer, more distinct clusters.

We also include in the model a repulsive prior on the component locations and a prior on the number of mixture components. Leveraging recent results for MCMC methods for repulsive mixtures \citep{beraha_2021_repmix_mcmc, cremaschi_2023_chaos}, we devise an MCMC algorithm that performs a trans-dimensional move via the use of a birth-and-death step. We find that the novel approach has an important implication on the implied prior on the number of clusters, inducing additional sparsity through an increase in the repulsive parameter. We explore the complex relationship between the repulsive parameters of the component locations and the component weights in a simulation study. While an increase in the former tends to increase the number of clusters, the latter induces a higher degree of sparsity. Thus, we leverage this inverse relationship by modelling the ratio of the two repulsive parameters, providing also a data-driven approach to prior elicitation. 

We apply the repulsive mixture model with the SIP prior to a biomedical dataset from the GUSTO cohort study, focusing on children's Z-BMI and eating behaviour. Our results demonstrate that the SIP prior, through its effect on mixture weights, can counterbalance the strong separating tendency introduced by repulsion on component locations. The joint use of repulsion in both locations and weights leads to well-separated, non-redundant clusters, thereby enhancing interpretability.

Future research includes leveraging the enhanced sparsity of the proposed prior in empirical applications, and further exploring its integration within repulsive mixture models. An additional avenue of investigation is the relationship between the repulsion parameter and the entropy of the prior, which may offer valuable insights into its capacity to balance cluster separation with component sparsity.

\clearpage

\bibliographystyle{plainnat}
\bibliography{main}






\section*{Supplementary Material}
\renewcommand{\thesubsection}{\Alph{subsection}}  
\subsection{Properties of the Selberg Dirichlet Distribution}
\label{supp:sdir_props}
Leveraging known results from the theory of Mehta integrals, we can derive several key properties of the Selberg Dirichlet distribution, including closed-form expressions for marginal and joint moments.

For a vector $\bm{w} = (w_1, \ldots, w_M)$ distributed according to the M-dimensional Selberg Dirichlet distribution and defined on the simplex, where $0 \leq w_m \leq 1$ for all $m$ and $\sum_{m = 1}^M w_m = 1$, the expectation of the product can be expressed as a ratio of normalising constants:
\[
\mathbb{E} \left\{ \prod_{i=1}^{M} w_i \right\} = \frac{D(\alpha+1,\gamma,M)}{D(\alpha,\gamma,M)}
\]
This simplifies to
\begin{align*}
\frac{D(\alpha+1, \gamma, M)}{D(\alpha, \gamma, M)} 
&= \frac{\Gamma(\alpha+1)}{\Gamma(\alpha)} 
\frac{\Gamma(M\alpha + \gamma (M - 1)(M - 2))}{\Gamma(M\alpha + M + \gamma (M - 1)(M - 2))}
 \prod_{j = 1}^{M - 1} \frac{\Gamma(\alpha + 1 + (j - 1)\gamma)}{\Gamma(\alpha + (j - 1)\gamma)} \\
&= \alpha 
\frac{\Gamma(M\alpha + \gamma (M - 1)(M - 2))}{\Gamma(M\alpha + M + \gamma (M - 1)(M - 2))} 
 \prod_{j = 1}^{M - 1} (\alpha + (j - 1)\gamma)
\end{align*}
The marginal expectation is given by the following ratio:
\[
\mathbb{E} \left\{ w_j \right\} = \frac{A(\alpha, \alpha+1, \gamma, M)}{D(\alpha,\gamma,M)}
\]
where \( A(\alpha, \beta, \gamma, M) \) is defined by the Mehta integral as
\[
A(\alpha, \beta, \gamma, M) = \frac{\Gamma(\beta)}{\Gamma(\alpha(M-1)+\beta+(M-1)(M-2)\gamma)}
\prod_{j=1}^{M-1} \frac{\Gamma(\alpha + (j-1)\gamma) \Gamma(1+j\gamma)}{\Gamma(1+\gamma)}
\]
After simplification, the marginal expectation reduces to
\begin{align*}
\mathbb{E} \left\{ w_j \right\} 
&= \frac{\Gamma(\alpha+1)}{\Gamma(\alpha(M-1)+\alpha + 1 +(M-1)(M-2)\gamma)}  
      \frac{\Gamma(M\alpha + \gamma (M - 1)(M-2)) }{ \Gamma(\alpha)} \\
&= \frac{\alpha}{\alpha M + (M - 1)(M - 2)\gamma}
\end{align*}
highlighting the influence of the repulsion parameter \( \gamma \) and the number of components \( M \) on the mean.
The expectation of higher-order joint moments admits the following expression:
\[
\mathbb{E} \left\{ \prod_{i=1}^{M} w_i^{k} \right\} = \frac{D(\alpha+k, \gamma, M)}{D(\alpha, \gamma, M)}
\]
while the second moment of any marginal \( w_j \) is given by
\[
\mathbb{E} \left\{ w_j^2 \right\} = \frac{A(\alpha, \alpha+2, \gamma, M)}{D(\alpha, \gamma, M)}
\]
which simplifies to
 \begin{align*}
     E\{w_i^2\} = 
\frac{\Gamma(\alpha+2)}{\Gamma(\alpha M + 2 +(M-1)(M-2)\gamma)} 
    \frac{\Gamma(M\alpha + \gamma (M - 1)(M-2)) }{ \Gamma(\alpha)} = \\
    \frac{\alpha(\alpha+1)}{(\alpha M + 1 +(M-1)(M-2)\gamma)(\alpha M + (M-1)(M-2)\gamma)}
\end{align*}
Furthermore, we can derive the variance of \( w_j \) using the expressions for the first and second moments. Define
\[
\eta = \alpha M + (M - 1)(M - 2)\gamma
\]
then
\begin{align*}
\mathbb{V} \left\{ w_j \right\} &= \mathbb{E} \left\{ w_j^2 \right\} - \left( \mathbb{E} \left\{ w_j \right\} \right)^2 \\
&= \frac{(\alpha+1)\alpha}{\eta(\eta+1)} - \left( \frac{\alpha}{\eta} \right)^2 \\
&= \frac{\alpha}{\eta} \left( \frac{1 - \frac{\alpha}{\eta}}{\eta + 1} \right)
\end{align*}

The derivative of the variance with respect to the repulsion parameter \( \gamma \) is negative for $M\geq 0$, implying that the variance is a decreasing function of $\gamma$: 
\begin{align*}
&\frac{d}{d\gamma} \mathbb{V} \left\{ w_j \right\} =\\
&-\frac{(M - 2)(M - 1) \alpha \left[ (M - 2)(M - 1) \gamma \left(2 (M - 2)(M - 1) \gamma + (4M - 3)\alpha + 1\right) + \alpha \left(M ((2M - 3)\alpha + 1) - 2\right) \right]}{(\eta)^3 (\eta + 1)^2}
\end{align*}
In general,
\[
\mathbb{E} \left\{ w_j^k \right\} = \frac{A(\alpha, \alpha+k, \gamma, M)}{D(\alpha, \gamma, M)}
\]
\subsection{Algorithm}
\label{supp:algorithm}
The following section provides details of the MCMC algorithm introduced in Section~\ref{sec:tiebreaker_mixture}. We follow recent literature \citep{argiento_infinity_2022} in differentiating between sets of parameters associated with allocated and non-allocated components, denoted by the subscripts $a$ and $na$, respectively. The total number of components becomes the sum of allocated and non-allocated components $M = M_a + M_{na}$.

In this section, the superscripts ``\(+\)'' and ``\(-\)'' are used to indicate proposed and current values, respectively, whenever a Metropolis--Hastings step is involved.  All components and dimensions are indexed by $1:M$ or $1:D$, respectively.

\subsubsection*{1 Updating the cluster allocations $\bm c$}
For every observation and all allocated and non-allocated components, we update the allocation variable $c_i$ with probability given by 
\begin{align*}
    P(c_i = m \mid \bm y, w_m, \bm \mu_{m,1:D}, \bm \Sigma_{m}) \propto w_m \mathcal{N}\left( \bm y_i \mid \bm \mu_{m,1:D}, \bm \Sigma_{m} \right) , \quad m=1, \ldots, M
\end{align*}
After all observations have been assigned to a component, the number of allocated and non-allocated components $M_a$ and $M_{na}$ are updated accordingly.
\subsubsection*{2 Updating the cluster-specific means $ \mu_{m, d}$ and covariance matrices $\bm \Sigma_m$}

For each dimension $d$ of all component means, we use independent $GE$ priors with equal repulsive parameters $\zeta$. The parameters \( \mu_{m,d} \) associated with allocated components are updated using their full conditional distributions via Metropolis--Hastings steps with Gaussian proposals and tuning parameter \( \sigma^2_{\mu} \):

\begin{align*}
r_{ \mu_{m, d}} &=  \frac{p(\mu^+_{m,d} \mid \bm{\mu}_{-m,d}, M, \zeta)}{p(\mu^-_{m,d} \mid \bm{\mu}_{-m,d}, M, \zeta)} 
\frac{\prod_{i:c_i = m}\mathcal{N}(\bm y_i \mid 
\bm \mu_{m, 1:D}^+ ,\bm \Sigma_m)}
{\prod_{i:c_i = m}\mathcal{N}(\bm y_i \mid 
\bm \mu_{m, 1:D}^- ,\bm \Sigma_m)} 
\frac{\mathcal{N}(\mu^+_{m,d} \mid \mu^-_{m,d}, \sigma^2_{\mu})}{\mathcal{N}(\mu^-_{m,d} \mid \mu^+_{m,d}, \sigma^2_{\mu})} \\
& \propto \frac{GE(\bm{\mu}_{1:M, d}^+ \mid  M, \zeta)}{GE(\bm{\mu}_{1:M, d}^- \mid  M, \zeta)}
\frac{\prod_{i:c_i = m}\mathcal{N}(\bm y_i \mid 
\bm \mu_{m, 1:D}^+,\bm \Sigma_m)}
{\prod_{i:c_i = m}\mathcal{N}(\bm y_i \mid 
\bm \mu_{m, 1:D}^-, \bm \Sigma_m)}\\
&= \frac{ e^{-\frac{\zeta}{2} \mu_{m,d}^{+,2}} |\triangle \bm \mu^+_{1:M,d}|^{\zeta} }{ e^{-\frac{\zeta}{2} \mu_{m,d}^{-,2}} |\triangle \bm \mu^-_{1:M,d}|^{\zeta} }
\frac{ \exp \left( -\frac{1}{2} \sum_{i:c_i = m} (\bm{y}_i - \bm{\mu}^+_{m,1:D})^\top \bm{\Sigma}_m^{-1} (\bm{y}_i - \bm{\mu}^+_{m,1:D}) \right) }{ \exp \left( -\frac{1}{2} \sum_{i:c_i = m} (\bm{y}_i - \bm{\mu}^-_{m,1:D})^\top \bm{\Sigma}_m^{-1} (\bm{y}_i - \bm{\mu}^-_{m,1:D}) \right)} \nonumber
\end{align*}
where we denote with $ \bm{\mu}_{-m,d}$ the vector of mean parameters of dimension $d$ excluding the $m$-th element.

Due to the conjugacy of the inverse-Wishart distribution, the posterior of $\bm \Sigma_m$ is an updated inverse-Wishart distribution:
\begin{align*}
&p(\bm{\Sigma}_m \mid \bm y,  \bm V_{0}, \nu_{0}) \sim IW( \bm V_{post}, \nu_{post})\\
&\bm V_{post} = n_m  \frac{\sum_{i:c_i = m}^N \bm y_i \bm y_i^T}{n_m} +\mathbf{V}_0, \quad \nu_{post} = n_m + \nu_0
\end{align*}

The parameters associated with non-allocated components are updated using their prior distributions.

\subsubsection*{3 Updating the repulsive weights $\bm w$}
Our version of the Selberg Dirichlet distribution, as defined in Equation~\ref{eq:sdir_prior}, requires equal concentration parameters. This not only makes it a more restrictive prior than the Dirichlet but prevents conjugacy with the categorical or multinomial distribution. One of the main appeals of using the Dirichlet distribution as prior for mixture weights is that its posterior is again a Dirichlet with concentration parameters incremented by the number of observations assigned to the respective cluster. A generalisation of the Selberg Dirichlet distribution, enabling varying concentration parameters, would therefore form a conjugate pair in a mixture model and enable a clear sampling strategy. We designate the resulting distribution as the \textit{generalized Selberg Dirichlet} (GSDir), which is defined as
\begin{definition}[generalized Selberg Dirichlet]
The M-dimensional generalized Selberg Dirichlet has a density given by:
\begin{align}
&GSDir(\bm{w} , \bm{\alpha}, \gamma, M) = 
\frac{ 1}{GD(\bm{\alpha}, \gamma, M)} 
\left( \prod_{m = 1}^M  w_m^{\alpha_m -1} \right) |\triangle \bm w|^{2\gamma} \label{eq:gsdir_1}\\
&GD(\bm{\alpha}, \gamma, M) = \int_0^{1} \dots \int_0^{1} \left( \prod_{m = 1}^M  w_m^{\alpha_m -1} \right) |\triangle \bm w|^{2\gamma} dw_1 \dots dw_M \label{eq:gsdir_2}
\end{align}
where $\bm \alpha > 0$, $\gamma \geq 0$ and $\sum_{m = 1}^M w_m = 1$.
\end{definition}
To the best of our knowledge, a normalising constant is known in closed form only when all concentration parameters $\alpha_m$ are equal. However, it is easy to see that Equation~\ref{eq:gsdir_2} is finite, by noting that both the product term and the pairwise differences are bounded between $0$ and $1$, thereby ensuring the existence of the distribution. 
\begin{figure}[tp]
\centering
     \includegraphics[width=.22\linewidth]{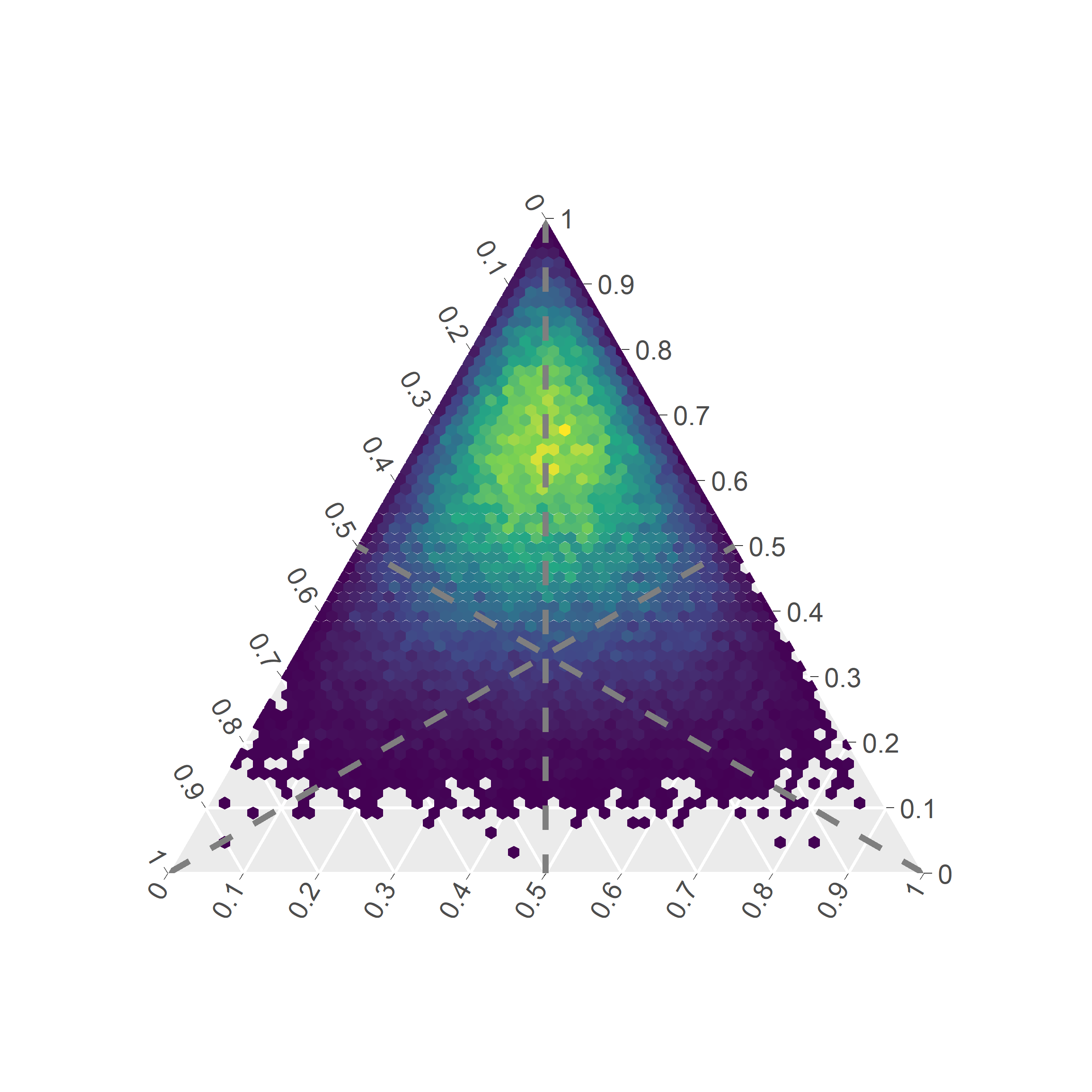}
    \includegraphics[width=.22\linewidth]{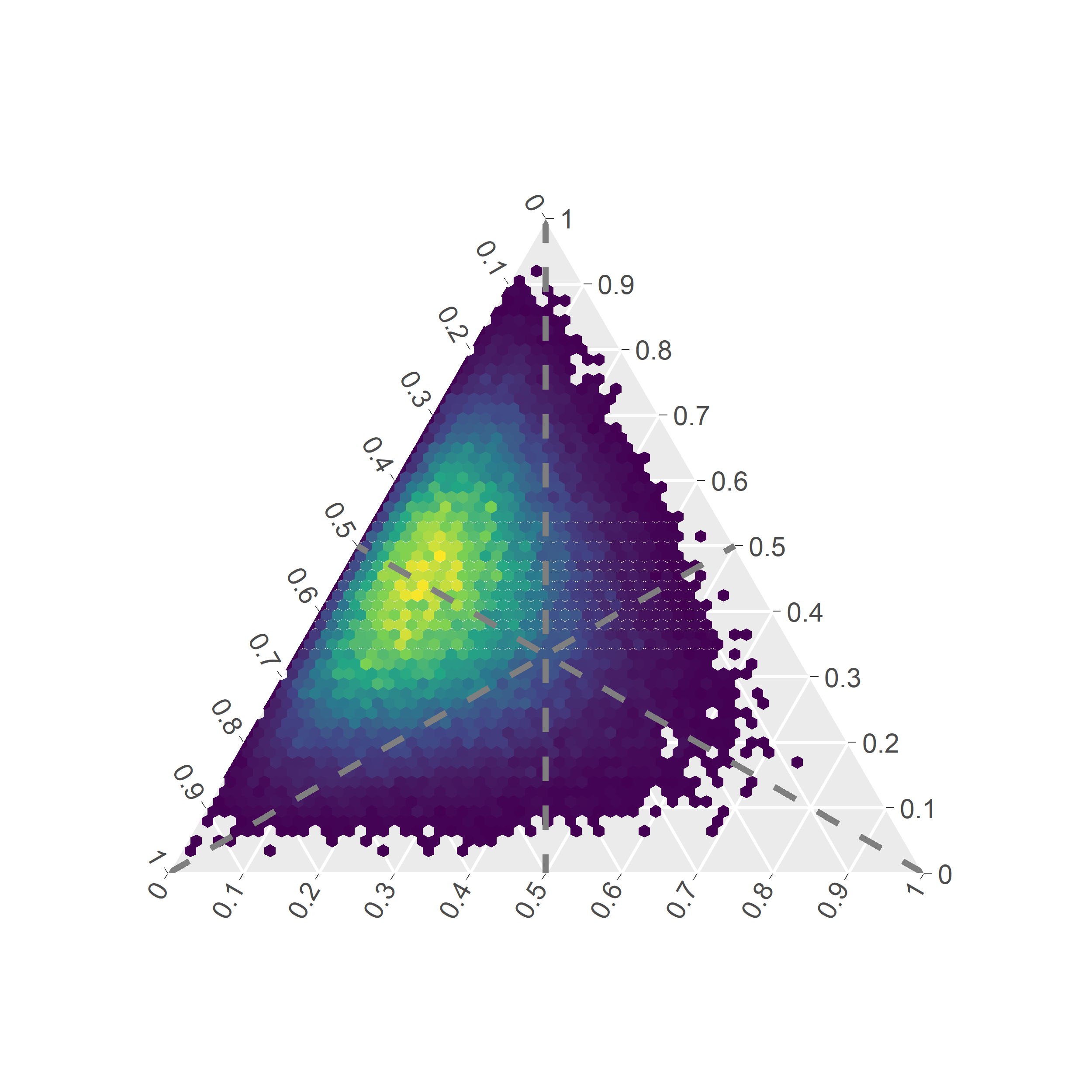}
    \includegraphics[width=.22\linewidth]{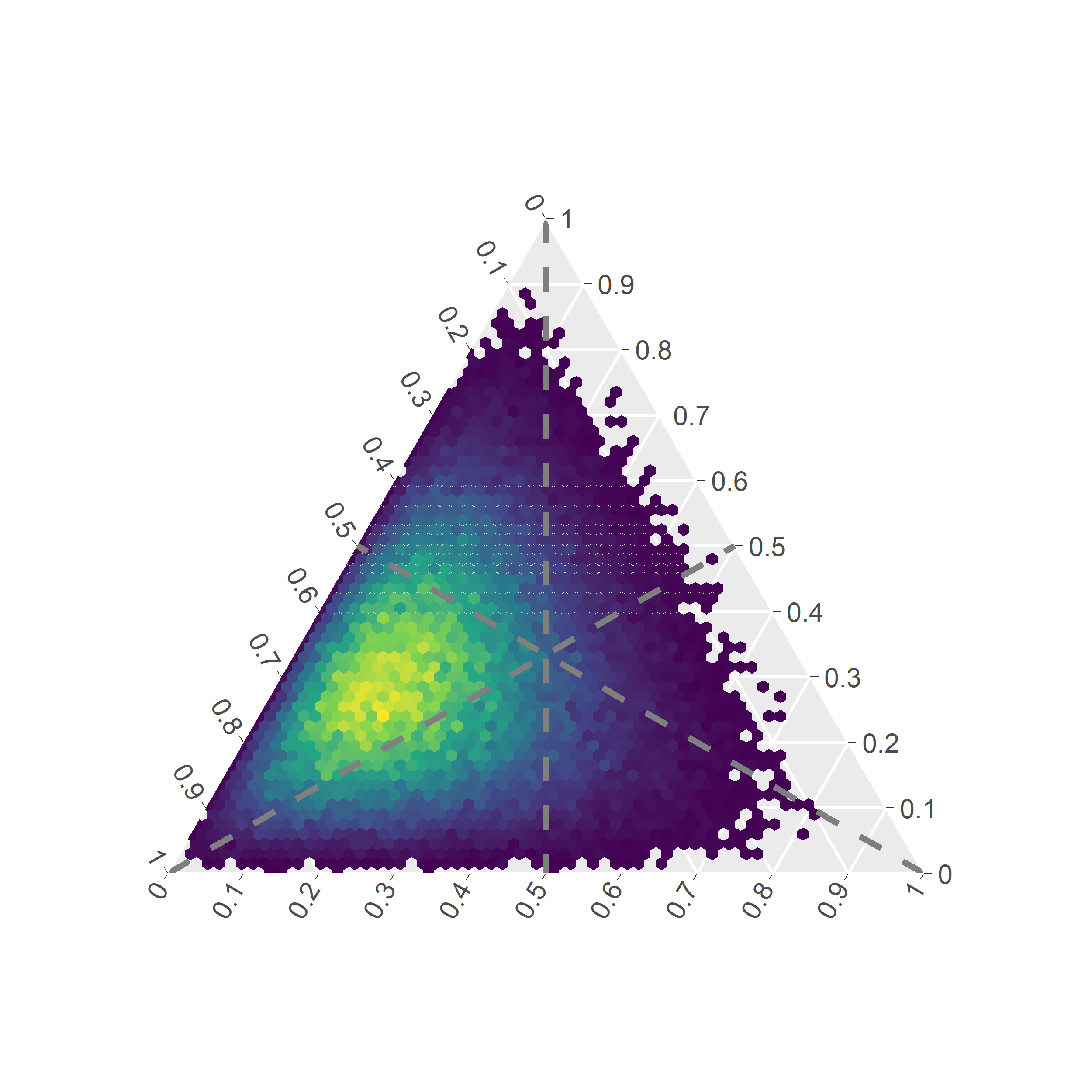}
    \includegraphics[width=.22\linewidth]{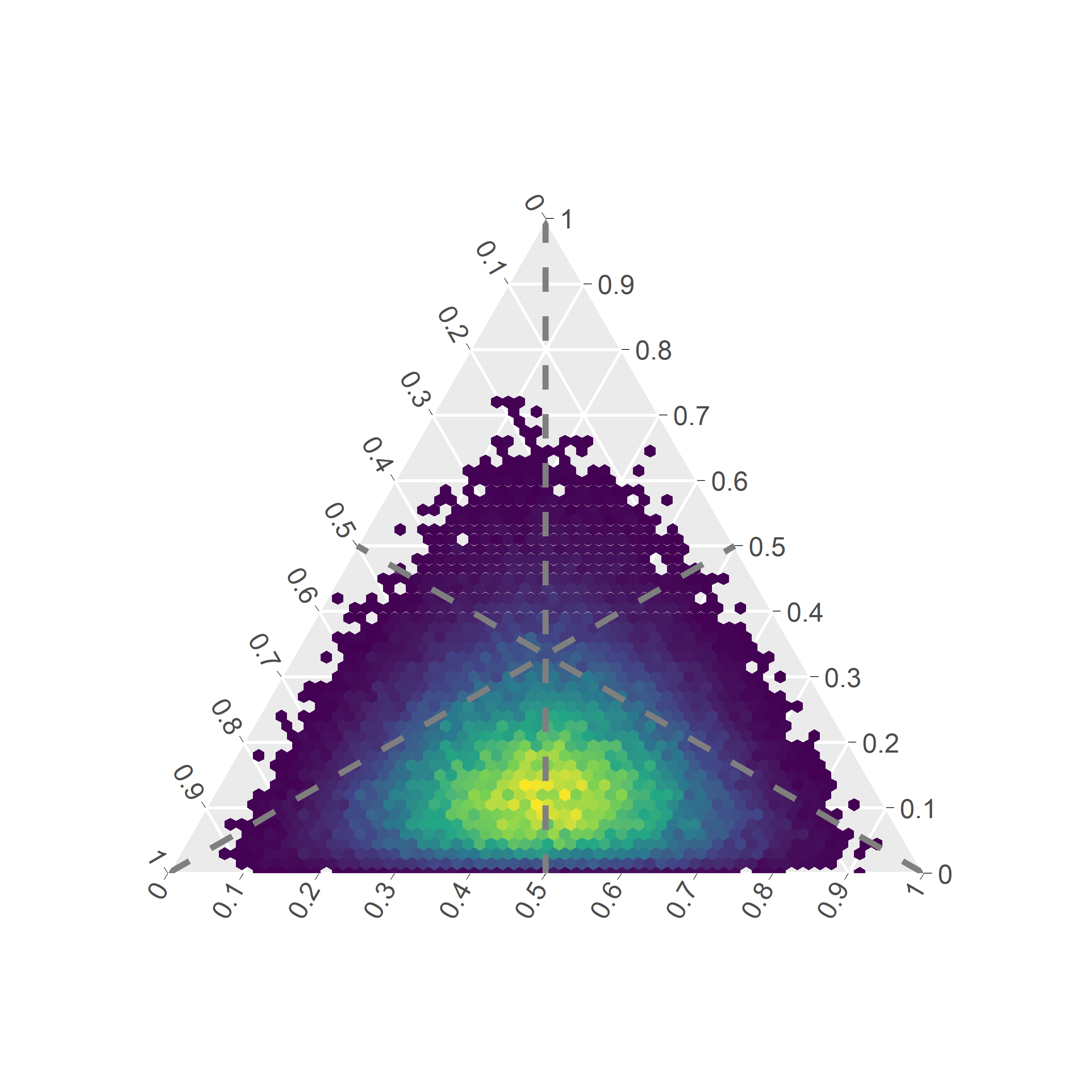}\\
   \includegraphics[width=.22\linewidth]{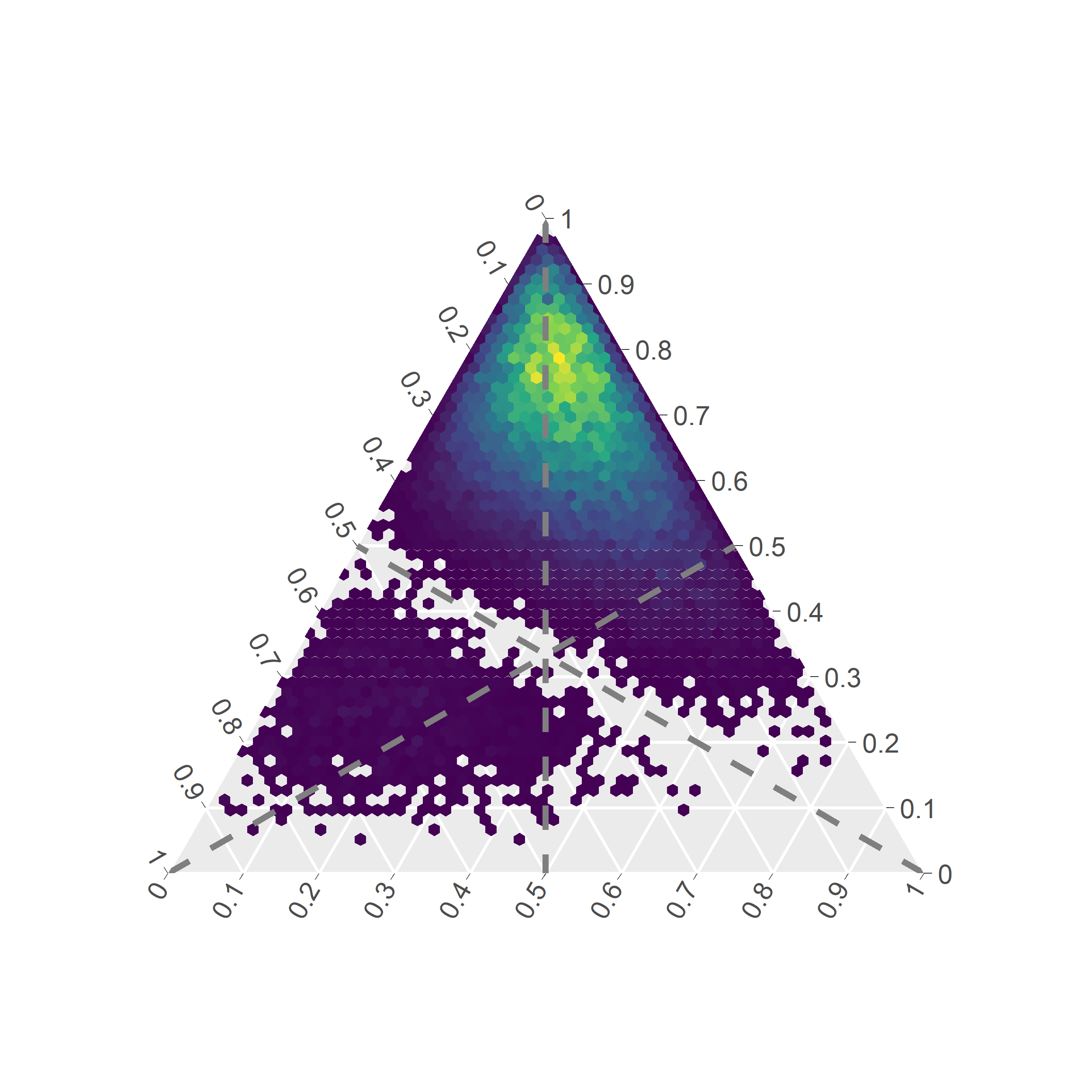}
   \includegraphics[width=.22\linewidth]{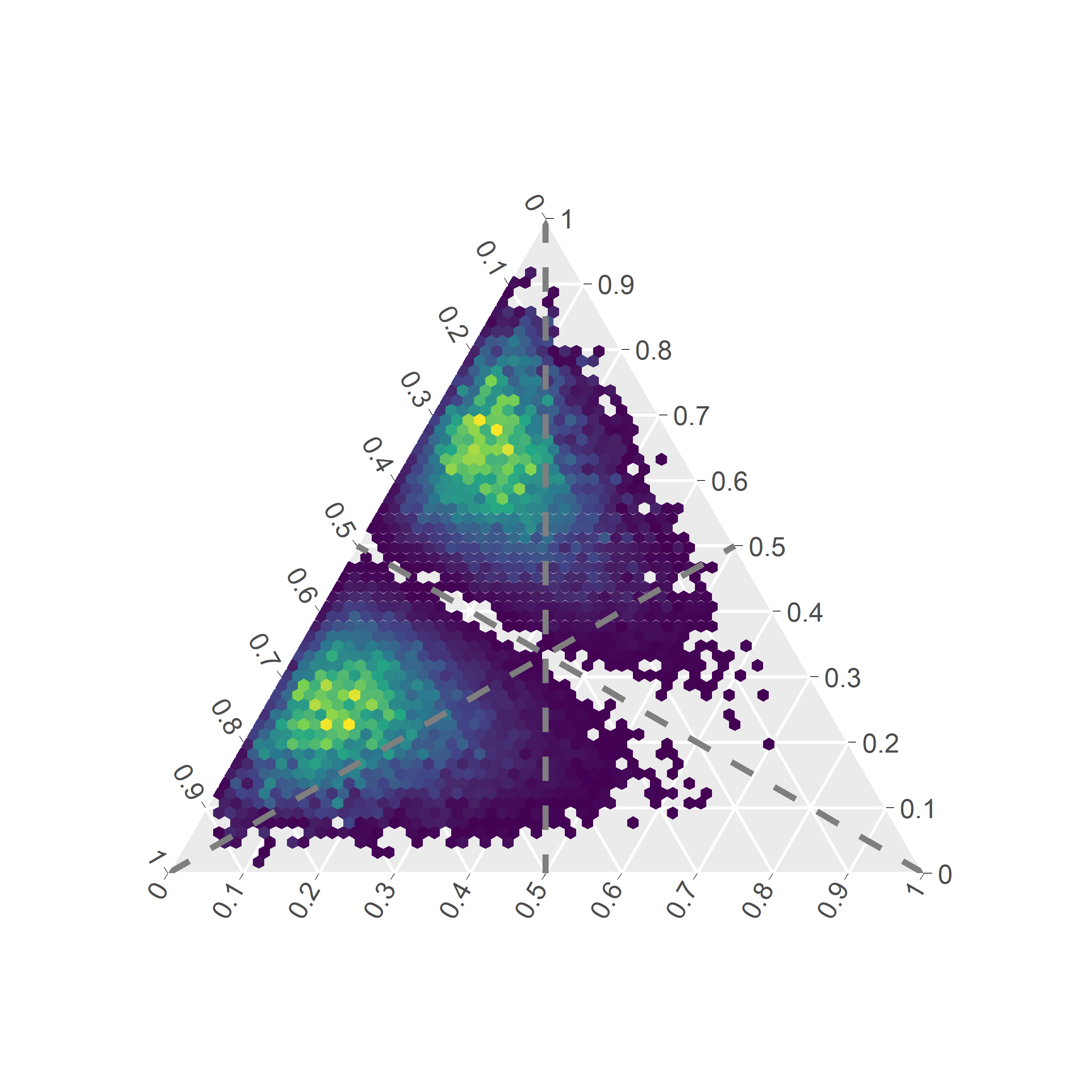}
   \includegraphics[width=.22\linewidth]{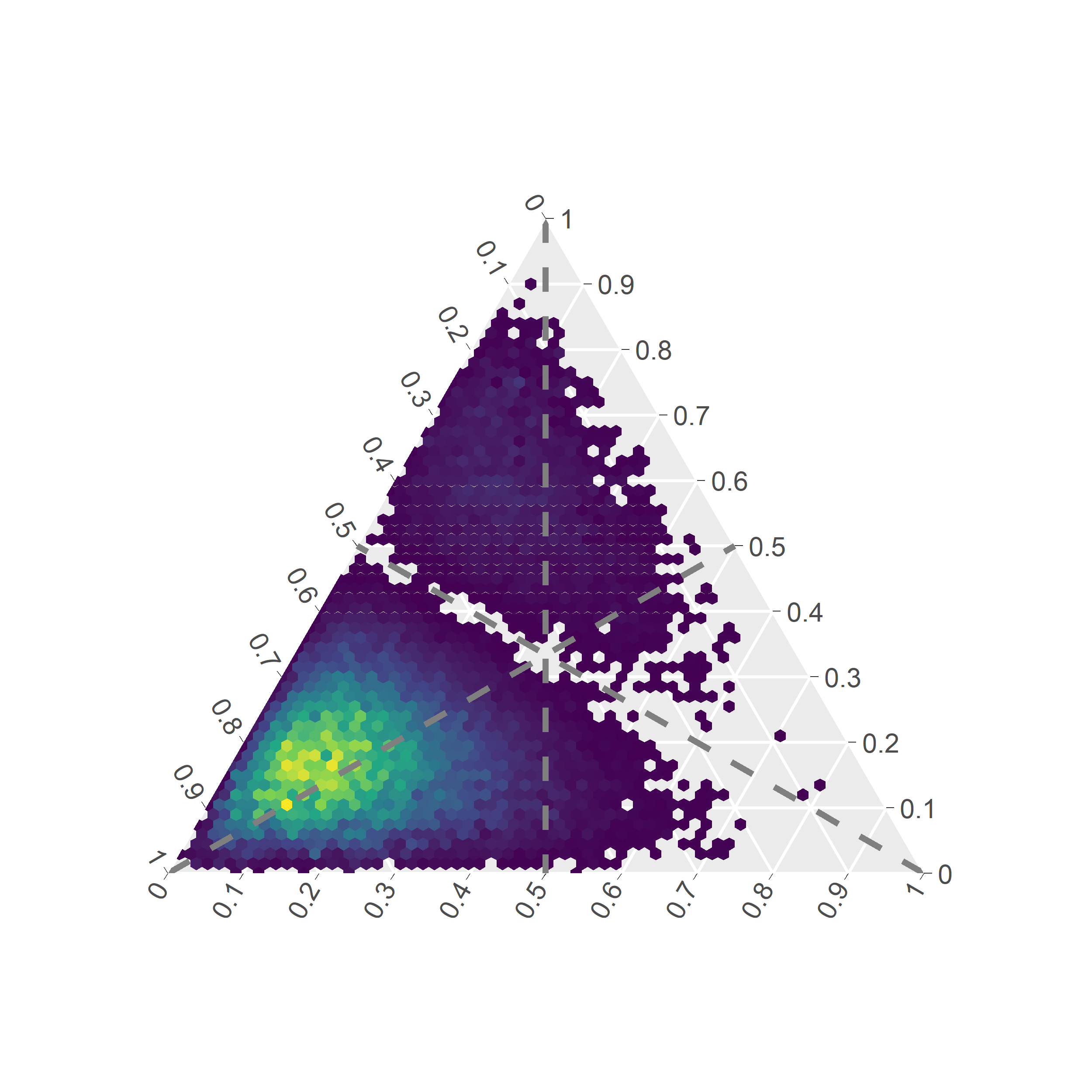}
    \includegraphics[width=.22\linewidth]{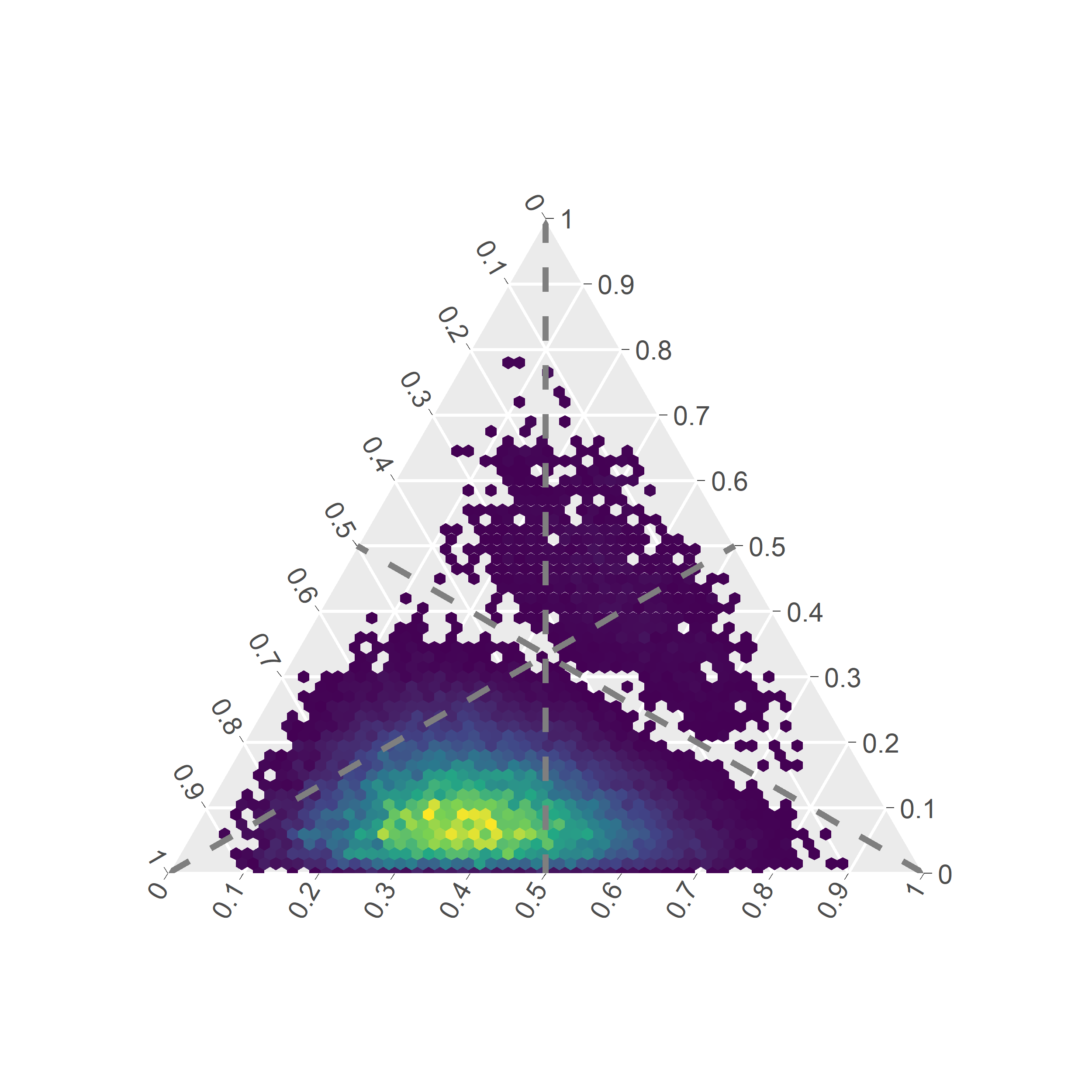}
\caption{Simulation-based densities of the Dirichlet and Selberg Dirichlet distribution for varying values of the concentration parameter $\alpha$ and $\gamma = 1$. From left to right: $\bm \alpha = (2,5,2),(5,5,2),(5,5,2),(5,3,2)$. The first row shows the Dirichlet distribution and the second the generalized Selberg Dirichlet with $\gamma = 1$. For illustration purposes, the colors are scaled separately for each plot.}
\label{fig:gsdir_prior}
\end{figure}

Figure~\ref{fig:gsdir_prior} presents the generalized Selberg Dirichlet as well as the Dirichlet distribution for $M = 3$, unequal $\alpha_m$ and $\gamma = 1$. We can observe that the repulsive property clearly divides the density into multiple modes.

The varying combinations of alpha parameters can be viewed as either an informed prior choice or a posterior update that directs preference toward a particular corner or edge of the simplex. For the standard Dirichlet distribution, this typically results in one distinct value and two similar ones. In contrast, the repulsive parameter of the Selberg Dirichlet further separates values from one another.

Employing the SIP prior for $\bm w$ with concentration parameter $\alpha_0$ results in the generalized Selberg Dirichlet in Equation~\ref{eq:gsdir_1} as full conditional distribution of the component weights, with concentration parameters incremented by the number of data points allocated to the corresponding cluster $ \alpha_{m, post} = \alpha_0 + n_m$. To circumvent the necessity of a known and closed-form normalising constant, we update $\bm w$ using an MH-step, as the normalising constants cancel out in the acceptance probabilities $r_{w}$. We use a standard Dirichlet distribution with concentration parameters $\bm \alpha_{ post} = (\alpha_{m, post}, \dots, \alpha_{M, post})$ as proposal density:
\begin{align*}
r_{w} &= \frac{GSDir(\bm{w}^+\mid \bm{\alpha}_{post}, \gamma, M ) }{GSDir(\bm{w}^-\mid \bm{\alpha}_{post}, \gamma, M ) } \frac{Dir(\bm{w}^- \mid \bm{\alpha}_{post})}{Dir(\bm{w}^+\mid \bm{\alpha}_{post})} \\
&= \frac{ \frac{1}{GD(\bm{\alpha}_{post}, \gamma, M)} \prod_{m = 1}^{M}  w_m^{+, \alpha_{m, post} - 1 } |\triangle\bm{w}^+|^{2\gamma} }{\frac{1}{GD(\bm{\alpha}_{post}, \gamma, M)} \prod_{m = 1}^{M}  w_m^{-,\alpha_{m, post} -1 } |\triangle\bm{w}^-|^{2\gamma} }
\frac{\frac{\Gamma(\sum_{m =1}^{M} \alpha_{m, post})}{\prod_{m =1}^{M}\Gamma(\alpha_{m, post})} \prod_{m =1}^{M}{w}_{m}^{-,\alpha_{m, post} - 1}}{\frac{\Gamma(\sum_{m =1}^{M} \alpha_{m, post})}{\prod_{m =1}^{M}\Gamma(\alpha_{m, post})} \prod_{m=1}^{M}{w}_{m}^{+,\alpha_{m, post} - 1}} \nonumber \\
&= \frac{ |\triangle \bm{w}^+|^{2\gamma} }{|\triangle \bm{w}^-|^{2\gamma}} \nonumber
\end{align*}
As we can see, the acceptance probabilities simplify to a ratio of the repulsive terms, implying that a set of weights with greater pairwise differences than the old will always be accepted, assuming $0 < \gamma < 1$.

\subsubsection*{4 Updating the repulsive parameter $\gamma$}
Depending on the desired level of control over the repulsion, this parameter can either be fixed or updated by placing a prior on it. For the latter option we choose the gamma distribution with shape and rate parameters $\alpha_{\gamma, 0}$ and $\beta_{\gamma, 0}$ respectively, denoted by $\Gammad{\alpha_{\gamma, 0}, \beta_{\gamma, 0}}$.
The update is then performed via MH using a log-normal distribution with tuning parameter $\sigma^2_{\gamma}$ as proposal density. The acceptance probabilities $r_{\gamma}$ read as
\begin{align*}
r_{\gamma} &= 
\frac{SDir(\bm{w}\mid  \alpha_0, M, \gamma^+) \Gammad{\gamma^+\mid \alpha_{\gamma, 0}, \beta_{\gamma, 0}}}{SDir(\bm{w} \mid \alpha_0, M, \gamma^-) \Gammad{\gamma^- \mid \alpha_{\gamma, 0}, \beta_{\gamma, 0}}}
\frac{\lognorm{\gamma^- \mid \gamma^+, \sigma^2_{\gamma}}}{\lognorm{\gamma^+ \mid \gamma^-, \sigma^2_{\gamma}}}\\
&=\frac{ \frac{1}{D( \alpha_0, \gamma^+, M)} |\triangle \bm{w} |^{2\gamma^+} \gamma^{+,\alpha_{\gamma, 0}} exp(- \beta_{\gamma,0} \gamma^{+})  }{  \frac{1}{D( \alpha_0, \gamma^-, M)} |\triangle \bm{w} |^{2\gamma^-} \gamma^{-,\alpha_{\gamma, 0}} exp(- \beta_{\gamma,0} \gamma^{-}) }  \nonumber
\end{align*}
If $\zeta$ is defined through the ratio $\rho$, acceptance probabilities are updated as presented below.
\subsubsection*{5 Update the repulsion parameter $\zeta$}
Similarly to $\gamma$, the repulsion parameter on the component locations can either be fixed or learned in a Bayesian way. We present two different options on how to update $\zeta$ in case it is chosen to be updated during the sampling process. The first is an update from its full conditional, similarly to the repulsive parameter $\gamma$. The second requires the definition of the ratio between the repulsive parameters of the Gaussian ensemble priors and the SIP prior.
\subsubsection*{5.a Using the full conditional $p(\zeta \mid \bm \mu, M)$ and an MH step}
Likewise, the repulsion parameters $\zeta$ may either be updated or held constant. Employing the same prior and proposal distribution as for $\gamma$, we obtain the acceptance rates as follows
\begin{align*}
r_{\zeta} &= 
\frac{GE(\bm{\mu}_{1:M, d} \mid M, \zeta^+) \, \Gammad{\zeta^+ \mid \alpha_{\zeta, 0}, \beta_{\zeta, 0}}}
     {GE(\bm{\mu}_{1:M, d} \mid M, \zeta^-) \, \Gammad{\zeta^- \mid \alpha_{\zeta, 0}, \beta_{\zeta, 0}}}
\cdot 
\frac{\lognorm{\zeta^- \mid \zeta^+, \sigma^2_{\gamma}}}
     {\lognorm{\zeta^+ \mid \zeta^-, \sigma^2_{\gamma}}} \\[1ex]
&=
\frac{  
\frac{1}{\mathcal{G}(M, \zeta^+)}  
\prod_{m = 1}^M \exp\left(-\frac{\zeta^+}{2} \mu^2_{m, d}\right) 
|\triangle \bm{\mu}_{1:M,d} |^{\zeta^+} \,
\zeta^{+,\alpha_{\zeta, 0}} \exp(- \beta_{\zeta,0} \zeta^+)  }
{ 
\frac{1}{\mathcal{G}(M, \zeta^-)}  
\prod_{m = 1}^M \exp\left(-\frac{\zeta^-}{2} \mu^2_{m, d}\right) 
|\triangle \bm{\mu}_{1:M,d} |^{\zeta^-} \,
\zeta^{-,\alpha_{\zeta, 0}} \exp(- \beta_{\zeta,0} \zeta^-) }
\nonumber
\end{align*}
\subsubsection*{5.b Updating the repulsive parameter $\zeta$ through the ratio $\rho = \zeta/ \gamma$}
The inverse nature of the relationship between $\gamma$ and $\zeta$ complicates suitable prior elicitation. To capture this relationship, we propose to fix the ratio $\rho = \zeta / \gamma$ and define $\zeta$ deterministically through $\gamma$. This approach alters the acceptance probabilities of $r_{\gamma}$, which now have to take into account the information from the component locations:
\begin{align*}
r_{\gamma} &= 
\frac{ \prod_{d=1}^{D} GE(\bm{\mu}_{1:M,d}\mid M, \zeta^+) SDir(\bm{w}\mid  \alpha_0, M, \gamma^+) \Gammad{\gamma^+ \mid \alpha_{\zeta, 0}, \beta_{\zeta, 0}} }{
\prod_{d=1}^{D} GE(\bm{\mu}_{1:M,d}\mid M, \zeta^-) SDir(\bm{w} \mid \alpha_0, M, \gamma^-) \Gammad{\gamma^- \mid \alpha_{\zeta, 0}, \beta_{\zeta, 0}} }
\frac{\lognorm{\gamma^- \mid \gamma^+, \sigma^2_{\gamma}}}{\lognorm{\gamma^+ \mid \gamma^-, \sigma^2_{\gamma}}} \\
&= \frac{ \prod_{d=1}^{D} \frac{1}{\mathcal{G}(M, \rho \gamma^+)}  \prod_{m = 1}^M \exp\left(-\frac{\rho \gamma^+}{2} \mu_{m,d}^2 \right) |\triangle \bm{\mu}_{1:M,d} |^{\rho \gamma^+} 
\frac{1}{D( {\alpha_0}, \gamma^+, M)} |\triangle \bm{w} |^{2 \gamma^+} \gamma^{+,\alpha_{\zeta, 0}} \exp(- \beta_{\zeta,0} \gamma^{+})  }{ 
\prod_{d=1}^{D} \frac{1}{\mathcal{G}(M, \rho \gamma^-)}  \prod_{m = 1}^M \exp\left(-\frac{\rho \gamma^-}{2} \mu_{m,d}^2 \right) |\triangle \bm{\mu}_{1:M,d} |^{\rho \gamma^-} 
\frac{1}{D( {\alpha_0}, \gamma^-, M)} |\triangle \bm{w} |^{2\gamma^-} \gamma^{-,\alpha_{\zeta, 0}} \exp(- \beta_{\zeta,0} \gamma^-)}
\end{align*}
Alternative parameterizations, such as defining \( \gamma \) deterministically via \( \zeta \) or modeling the product \( \gamma \zeta \), were also explored. However, the formulation presented here yielded the most consistent results in practice.
\subsubsection*{6 Update $M_{na}$, $\bm c$, $\bm w$ and $\bm \mu$ using a birth-and-death step}

Unless the number of non-allocated components $M_{na} = 0$, in which case a birth-move is performed per default, we choose between a birth or death move with respective probabilities $q$ and $1 - q$. Subsequently, we compute the respective acceptance probabilities, which are expressed as the product of the ratios of the full conditionals, birth-and-death probabilities $\frac{q}{1 - q}$ as well as proposal densities, where we denote with $b(\cdot)$ and $d(\cdot)$ the birth and death proposal densities.

Since $\bm w$ is a compositional vector, adding a new component changes the weights of all existing components. For this reason, we denote the $(M + 1)$-dimensional vector containing the novel weight by $\Tilde{\bm w}$, proposed from a Dirichlet distribution with concentration parameters given by $\bm \alpha_{post}$ for the allocated and $\alpha_0$ for the non-allocated components and denoted by $ \Tilde{\bm\alpha}$. Similarly, we denote by $\bm \hat{w}$ the $(M-1)$-dimensional vector excluding the component chosen to be closed in a death move and by $\hat{\bm \alpha}$ the vector of concentration parameters excluding the concentration parameter corresponding to the component chosen to be closed. The cluster allocations are included through the categorical distribution, denoted by $\mathcal{C}$.

In case of a birth move, the birth and death proposals read as
\begin{align*}
    b(\tilde{\bm w}, \bm \mu_{M+1, 1:D}) &= Dir(\tilde{\bm w} \mid \Tilde{\bm\alpha}) \times \mathcal{N}\left(\bm 0, \bm{\Sigma}_{\mu, 0} \right)\\
    d(\bm w) &= \frac{1}{M_{na} + 1} \times Dir(\bm w \mid \bm{\alpha}_{post})
\end{align*}
where \( \bm{0} \) denotes a zero vector of dimension \( D \), \( \bm{\Sigma}_{\mu, 0} \) a diagonal matrix with entries equal to \( 1/\zeta \) and $\frac{1}{M_{na} + 1}$ the uniform probability of choosing a non-allocated component to be closed. In case of a death move, the proposals become
\begin{align*}
b(\bm {w}) &= Dir(\bm {w} \mid \bm \alpha_{post})\\
d(\hat{\bm{w}}) &= \frac{1}{M_{na} - 1} \times Dir(\hat{\bm{w}} \mid \hat{\bm \alpha})
\end{align*}
The acceptance probabilities corresponding to a birth move are $A_b = \min(1, A_b^*)$ with
\begin{align*}
A_b^* &= 
\frac{ Poi_1(M + 1 \mid \lambda) \times
\mathcal{C}( \bm c \mid \Tilde{\bm{w}} ) \times
SDir(\Tilde{\bm{w}} \mid M + 1, \alpha_0, \gamma) 
}{
Poi_1(M \mid \lambda) \times 
\mathcal{C}(\bm c \mid \bm{w} )\times 
SDir(\bm{w}\mid  M , \alpha_0, \gamma) 
}  \nonumber\\
& \times
\frac{\prod_{d = 1}^D GE(\bm{\mu}_{1:M + 1, d} \mid M +1, \zeta)}{
\prod_{d = 1}^D GE(\bm{\mu}_{1:M, d}\mid M, \zeta) } \times \frac{ 1 - q}{q} \times 
\frac{d(\bm w)}{ b( \tilde{\bm w}, \bm \mu_{M+1, 1:D} )} \nonumber \\
&= \frac{\lambda}{M}  \times 
\frac{\frac{ 1}{D(\alpha_0, \gamma, M + 1)}  |\triangle \Tilde{\bm w}|^{2\gamma}}{\frac{ 1}{D(\alpha_0, \gamma, M )} |\triangle \bm w|^{2\gamma}}  \times 
 \frac{\prod_{d = 1}^D \mathcal{G}(M, \zeta) \prod_{m = 1}^{M} \left| {\mu}_{M + 1, d} -{\mu}_{m,d} \right|^{\zeta}}{ \mathcal{G}(M+1, \zeta)} \nonumber \\
& \times \frac{1 - q}{q} \times
\frac{1}{M^{(na) }+ 1}
\frac{ \Gamma(\tilde{\alpha}_{M + 1}) \Gamma(\sum_{m =1}^{M} \alpha_{m, post}) }{\Gamma(\sum_{m =1}^{M + 1} \tilde{\alpha}_{m})} 
(2\pi)^{D/2} |\bm{\Sigma}_{\mu, 0}|^{1/2}
\end{align*}
The corresponding acceptance rates for a Death-move are $A_d = \min(1, A_d^*)$ with
\begin{align*}
    A_d &= 
    \frac{ Poi_1(M - 1 \mid \lambda) \times
    \mathcal{C}( \bm c \mid \hat{\bm{w}}) \times 
    SDir(\hat{\bm{w}} \mid M - 1, \alpha_0, \gamma) 
    }{
    Poi_1(M \mid \lambda) \times 
    \mathcal{C}(\bm c \mid \bm{w}) \times
    SDir(\bm{w}\mid  M, \alpha_{0}, \gamma) 
    }  \\
   & \times
   \frac{\prod_{d = 1}^D  GE(\hat{\bm{\mu}}_{1:(M-1), d} \mid M - 1, \zeta) }{
    \prod_{d = 1}^D GE(\bm{\mu}_{1:M, d}\mid M , \zeta)  } \times
   \frac{ q}{1 - q} \times 
    \frac{b(\bm{w} )}{ d(\hat{\bm{w}})}  \\
 &= \frac{ M - 1  }{ \lambda }  \times 
\frac{\frac{ 1}{D(\alpha_0, \gamma, M - 1)} {|\triangle \hat{\bm w}|^{2\gamma}}}{\frac{ 1}{D(\alpha_0,\gamma, M )} {|\triangle \bm w|^{2\gamma}}}  \times
  \frac{  \mathcal{G}(M, \zeta) 
}{\prod_{d = 1}^D \mathcal{G}(M - 1, \zeta) { \prod_{m = 1}^{M-1} \left|  \mu_{M, d} -  \mu_{m, d} \right|^{\zeta}} }
 \nonumber \\
& \times \frac{q}{1 - q}  \times  \frac{ \Gamma(\sum_{m =1}^{M}  \alpha_{m, post}) }{  \Gamma( \hat{\alpha}_{M}) \Gamma(\sum_{m =1}^{M - 1}  \hat{\alpha}_{ m})} 
M^{(na)}
\end{align*}
The MCMC algorithm for the SIP mixture with a fixed number of components results as a special case of the presented algorithm, with no necessity for step 6 of the algorithm.

\subsection{Proof of Theorem~\ref{thm:ldp}}
\label{supp:ldp}


We denote by $S_M$ the collection $\{ t \in \mathbb{R}^M \, \big| \, 0 \leq t_i \leq 1, \sum_{i=1}^M t_i \leq 1  \}$ and by $\nu_M$ the measure induced on $S_M$ by the $SDir$ distribution, where we omit the $(M+1)$-th component. We define the empirical measure,
\[ E_M = \frac{1}{M} \sum_{i = 1}^M \delta_{w_i} \]
which is a (random) measure on $[0,1]$. More generally, for $t \in S_M$, we define the discrete measure
\[ \kappa_{M,t} = \frac{1}{M}\sum_{i = 1}^M \delta_{t_i} \]
Then, for any Borel subset $\Gamma \subset \mathcal{M}([0,1])$, we define
\[ E_M(\Gamma) := \nu_M \left( \{ t \in S_M : \kappa_t \in \Gamma  \} \right) \]
For convenience, we take $\alpha - 1 = \beta$ and assume that $\frac{M}{\beta} \to a$, and define the functions
\[ F(x,y) := -2a^2\gamma\log|x-y| - a\left( \log x + \log y \right)  \]
and
\[ F_M(x,y) := -2\gamma\frac{M^2}{\beta^2}\log|x-y| - \frac{M}{\beta}\left( \log x + \log y \right) \]
For $R>0$, we also define truncated versions of these functions as
\[ F_R(x,y) = \min\{ F(x,y), R \} \] and \[F_{M,R}(x,y) = \min\{ F_M(x,y), R \}  \]
Note that $F_{M,R}$ converges to $F_R$ uniformly. Moreover, for each $R$, $F_R$ is bounded and continuous. Then, since $F_R \to F$ monotonically, we have that
\[ \iint F(x,y) d\mu(x)d\mu(y) = -2a^2\gamma\iint \log|x-y| d\mu(x)d\mu(y) - 2a\int \log x d\mu(x)  \]
is a well-defined and lower semi-continuous functional on $\mathcal{M}([0,1])$; the same holds for $F_M$. 
Furthermore, define $\mu_0, \mu_M \in \mathcal{M}([0,1])$ so that
\[ \iint F(x,y) d\mu_0(x)d\mu_0(y) = \inf_{\mu \in \mathcal{M}([0,1])} \iint F(x,y) d\mu(x)d\mu(y) \]
and
\[ \iint F_M(x,y) d\mu_M(x)d\mu_M(y) = \inf_{\mu \in \mathcal{M}([0,1])} \iint F_M(x,y) d\mu(x)d\mu(y) \]
The measures $\mu_0$ and $\mu_n$ are guaranteed to exist and have compact support by \cite{totik1994weighted}. Then, we obtain the map $I: \mathcal{M}([0,1]) \to \mathbb{R}$ defined by
\[ I(\mu) = -2a^2\gamma\iint \log|x-y| d\mu(x)d\mu(y) - 2a\int \log x d\mu(x)  + B \]
where
\[ B = \lim_{M \to \infty} \frac{1}{M^2} \log B_M \]
\begin{lemma}
The sequence $\{\mu_M\}$ is tight and
\[ \iint F(x,y) d\mu_0(x)d\mu_0(y) \leq \liminf_{M \to \infty} \iint F_M(x,y) d\mu_M(x)d\mu_M(y) \]
\end{lemma}

\begin{proof}
First, since the two-dimensional simplex is compact, we can immediately conclude that the sequence is tight. Then, there exists a sub-sequence $\{\mu_{M_k}\}$ that converges weakly to some $\hat{\mu} \in \mathcal{M}([0,1])$ and 
\begin{equation*}
\begin{gathered}
    \lim_{k \to \infty} \iint F_{M_k}(x,y) d\mu_{M_k}(x) d\mu_{M_k} (y) = 
    \liminf_{M \to \infty} \iint F_{M}(x,y) d\mu_{M}(x) d\mu_{M}(y)
\end{gathered}
\end{equation*}
It follows that
\begin{align*}
\iint F(x,y) d\mu_0(x) d\mu_0(y) & \leq \iint F(x,y) d\hat{\mu}(x) d\hat{\mu}(y) \\ & = \sup_{R > 0} \iint F_R(x,y) d\hat{\mu}(x) d\hat{\mu}(y) \\ & = \sup_{R > 0} \lim_{M \to \infty} \iint F_{R,M}(x,y) d\mu_{M_i}(x)d\mu_{M_i}(y) \\ & \leq \lim_{M \to \infty} \iint F_{R,M}(x,y) d\mu_{M_i}(x)d\mu_{M_i}(y)
\end{align*}
which yields the desired inequality.
\end{proof}
The following argument establishes only a weak large deviations principle (LDP). To obtain a full LDP in general, one must additionally verify the exponential tightness of the sequence of measures \citep{RAS2015LDPgibbs} However, in our setting, the compactness of the underlying space guarantees this property automatically, so no further conditions are needed.
\begin{lemma}
The joint probability density function can be rewritten as
\[ SDir = B_M^{-1} \exp \left[ -\frac{\beta^2}{M^2} \sum_{1 \leq i < j \leq M} F_M(w_i,w_j)  \right] \exp \left[ \frac{\beta}{M} \sum_{i = 1}^M \log w_i  \right] \]
\end{lemma}

\begin{proof}
We ignore the normalizing constant for the moment and compute 
\begin{align*}
\left( \prod_{i = 1}^M w_i^\beta \right)  \left( \prod_{i < j} |w_i - w_j|^{2\gamma} \right) & = \exp\left\{\log\left[  \left( \prod_{i = 1}^M w_i^\beta \right)  \left( \prod_{i < j} |w_i - w_j|^{2\gamma} \right) \right] \right\} \\ & = \exp \left[ \sum_{i = 1}^M \beta\log w_j + 2\gamma \sum_{i < j} \log|w_i - w_j| \right] \\
& = \exp \left[ \frac{M+1-1}{M} \sum_{i = 1}^M \beta\log w_j + 2\gamma \sum_{i < j} \log|w_i - w_j| \right] \\
& = \exp \left[ \frac{M-1}{M} \sum_{i = 1}^M \beta\log w_j + 2\gamma \sum_{i < j} \log|w_i - w_j| + \frac{1}{M} \sum_{i = 1}^M \beta\log w_j \right] \\
& = \exp \left[ \frac{\beta}{M} \sum_{i < j} \log w_i + \log w_j + 2\gamma \sum_{i < j} \log|w_i - w_j| + \frac{\beta}{M} \sum_{i = 1}^M \log w_j \right] \\
& = \exp \left[ \frac{\beta}{M} \sum_{i < j} \left( \log w_i + \log w_j + 2\gamma \log|w_i - w_j| \right) \right]  \exp \left[ \frac{\beta}{M} \sum_{i = 1}^M \log w_j \right] \\
& = \exp \left[ -\frac{\beta^2}{M^2} \sum_{i < j} F_M(w_i,w_j) \right] \exp \left[ \frac{\beta}{M} \sum_{i = 1}^M \log w_j \right]
\end{align*}
\end{proof}

The next two lemmas will be used to establish the large deviation upper bound.
\begin{lemma}
\[ \limsup_{M \to \infty} \frac{1}{M^2} \log B_M \leq - \iint F(x,y) d\mu_0(x)d\mu_0(y) \]
\end{lemma}

\begin{proof}
Using Lemma 2, we estimate
\begin{align*}
B_M & = \int \cdots \int \exp \left[ \frac{\beta^2}{M^2} \sum_{i < j} F_M(w_i,w_j) \right] \exp \left[ \frac{\beta}{M} \sum_{i = 1}^M \log w_j \right] dw_1 \dots dw_M \\
& \leq \int \cdots \int \exp \left[ \frac{\beta^2}{M^2} \sum_{i < j} F_M(w_i,w_j) \right] dw_1 \dots dw_M  \left( \int \exp \left[ \frac{\beta}{M} \sum_{i = 1}^M \log w_j \right] dx \right)^M  \\
& \leq \exp \left[ -\beta^2 \iint F_M(x,y) d\mu_M(x)d\mu_M(y) \right]\left( \int \exp \left[ \frac{\beta}{M} \sum_{i = 1}^M \log w_j \right] dx \right)^M
\end{align*}
Then, since
\[ \sup_{M \geq 1} \int \exp \left[ \frac{\beta}{M} \sum_{i = 1}^M \log w_j \right] dx < \infty  \]
we arrive at
\begin{align*} 
\limsup_{M \to \infty} \frac{1}{M^2} \log B_M & \leq -\liminf_{M \to \infty}  \iint F_M(x,y) d\mu_M(x)d\mu_M(y) \\
& \leq - \iint F(x,y) d\mu_0(x)d\mu_0(y)
\end{align*}
\end{proof}

\begin{lemma}
For every $\mu \in \mathcal{M}([0,1])$, we have
\[ \inf_G \left[ \limsup_{M \to \infty} \frac{1}{M^2} \log L_p(G) \right] \leq - \iint F(x,y) d\mu(x)d\mu(y) - \liminf_{M \to \infty} \frac{1}{M^2} \log B_M  \]
where $G$ runs over neighborhoods of $\mu$ in the weak topology.
\end{lemma}

\begin{proof}
Let $\mu \in \mathcal{M}([0,1])$ and let $G$ be a neighbourhood of $\mu$. For $t \in S_M$, define $D_{t,M}$ to be the $M \times M$ diagonal matrix whose entries are given by $t$. Define $\tilde{G} := \{ t \in S_M : \mu_{D_{t,M}} \in G \}$. Then, letting $\nu$ denote the measure corresponding to the density $SDir$, we have
\[ \mu_{E_M}(G) = \nu(\tilde{G}) \]
Now, note that $\mu_{E_M} \otimes \mu_{E_M}\left( \{x=y\} \right) = \frac{1}{M}$, from which we see
\begin{equation*}
    \begin{gathered}
        \iint f_{R,N}(x,y) d\mu_{E_M}(x) d\mu_{E_M}(y) =
        \iint\limits_{x \not = y} f_{R,N}(x,y) d\mu_{E_M}(x) d\mu_{E_M}(y) + \frac{R}{M}
    \end{gathered}
\end{equation*}
Then, rewriting the density as before, we obtain
\begin{align*} 
&\mu_{E_M}(G) = \nu(\tilde{G}) \leq B_M^{-1} \left( \int \exp \left[ \frac{\beta}{M} \sum_{i = 1}^M \log w_j \right] dx \right)^M \exp \left( - M^2 \inf_{\sigma \in G} \iint\limits_{x \not= y} F_{R,M} (x,y) d\sigma(x) d\sigma(y)  + MR \right)
\end{align*}
for any $R > 0$. Moreover, we know that
\begin{equation*}
    \begin{gathered}
        \lim_{M \to \infty} \left( \inf_{\sigma \in G} \iint F_{R,M}(x,y) d\sigma(x)d\sigma(y) \right) =
        \inf_{\sigma \in G} \iint F_R(x,y) d\sigma(x) d\sigma(y)
    \end{gathered}
\end{equation*}
Hence
\begin{equation*}
    \begin{gathered}
        \limsup_{M \to \infty} \frac{1}{M^2} \log \mu_{E_M}(G) \leq
        - \inf_{\sigma \in G} \iint F_R(x,y) d\sigma(x) d\sigma(y) - \liminf_{M \to \infty} \frac{1}{M^2} B^{-1}_M
    \end{gathered}
\end{equation*}
Since $F_R(x,y)$ is bounded and continuous, it defines a continuous functional so that
\begin{equation*}
    \begin{gathered}
        \inf_{G} \left( \limsup_{M \to \infty} \frac{1}{M^2} \log (\mu_{E_M}(G)) \right) \leq  -
        \iint F_R(x,y) d\mu(x) d\mu(y) - \liminf_{M \to \infty} \frac{1}{M^2} \log B_M
    \end{gathered}
\end{equation*}
Taking the limit as $R \to \infty$ and applying the monotone convergence theorem yield the desired result.
\end{proof}

Finally, we establish the large deviation lower bound.

\begin{lemma}
\[ \liminf_{M \to \infty} \frac{1}{n^2} \log B_M \geq - \iint F(x,y) d\mu_0(x)d\mu_0(y) \]
and, for every $\mu \in \mathcal{M}(\mathbb{R}^+)$, 
\[ \inf_G \left[ \liminf_{M \to \infty} \frac{1}{n^2} \log L_p(G) \right] \geq - \iint F(x,y) d\mu(x)d\mu(y) - \limsup_{M \to \infty} \frac{1}{M^2} \log B_M  \]
where $G$ runs over neighborhoods of $\mu$ in the weak topology.
\end{lemma}

\begin{proof}
Without loss of generality, we may assume that $\mu$ has a continuous density $f$ on $[0,1]$. Then, there exists a $\epsilon > 0$ such that $\epsilon \leq f(x) \leq \frac{1}{\epsilon}$ for $x\in[0,1]$. Next, for each $N$, define constants 
\[ 0 = s_{0,N} < r_{1,N} < s_{1,N} < \cdots < r_{M,N} < s_{M,N} = 1 \]
such that
\[ \int\limits_0^{r_{i,N}} f(x) dx = \frac{i - 1/2}{M} \hspace{1pc} \mbox{and} \hspace{1pc} \int\limits_0^{s_{i,N}} f(x) dx = \frac{i}{M} \]
Then we have
\[ \frac{\epsilon}{2M} \leq s_{i,N} - r_{i,M} \leq \frac{1}{2M\epsilon} \]
Now, define
\[ D_N = \left\{ (t_1, \dots, t_M) \in \mathbb{R}^M : r_{i,N} \leq t_i \leq s_{i,N} \right\} \]
Using the same notation from the previous lemma, for any neighbourhood $G$ of $\mu$ we can choose $N$ large enough so that $D_N \subset \tilde{G}$. Thus,
\begin{align*} 
\mu_{E_M}(G) &= \nu_M(\tilde{G}) \geq \nu_M(D_N)  \\
&= B_M^{-1} \int\cdots\int\limits_{D_N} \left( \prod_{k = 1}^M  t_k^\beta \right) \prod_{i < j} (t_i - t_j)^{2\gamma} dt_1\dots dt_M \\
& \geq B_{M}^{-1} \left( \frac{\gamma}{2M} \right)^M \prod_{i = 1}^M r_{i,N}^{\beta} \prod_{i < j} 	d_{i,j,N}
\end{align*}
where $d_{i,j,N} = \min\{ |x-y| : r_{i,N} \leq x \leq s_{i,N}, r_{j,N} \leq y \leq s_{j,N}\}$. Now, we observe that
\begin{equation*}
    \begin{gathered}
        \lim_{M \to \infty} \sum_{i = 1}^M \frac{M}{\beta} \log r_{i,N} = a \int \log x d\mu(x)
    \end{gathered}
\end{equation*}
and
\begin{equation*}
    \lim_{M \to \infty} \frac{1}{M^2} \sum_{i < j} \log(r_{j,N} - s_{i,N}) = 2a^2\gamma \int \log|x-y| d\mu(x)d\mu(y)
\end{equation*}
Thus,
\begin{equation*}
    \begin{gathered}
        \limsup_{M \to \infty} \frac{1}{M^2} \log \mu_{E_M}(G) \geq 
        - \iint F(x,y) d\mu(x)d\mu(y) - \liminf_{M \to \infty} \frac{1}{M^2} \log B_M
    \end{gathered}
\end{equation*}
After taking the infimum over $\mu$, this implies
\[ \liminf_{M \to \infty} \frac{1}{M^2} \log B_M \geq - \iint F(x,y) d\mu_0(x)d\mu_0(y) \]
Moreover, we have
\begin{equation*}
    \begin{gathered}
        \liminf_{N \to \infty} \log \mu_{E_M}(G) \geq 
        - \iint F(x,y)d\mu(x)d\mu(y) - \limsup_{M \to \infty} \frac{1}{M^2} \log B_M
    \end{gathered}
\end{equation*}
\end{proof}

\subsection{Cluster Differences in Maternal Characteristics}
\label{supp:plots}

To further characterise the clusters identified by minimising the Binder loss function in the application to the GUSTO data, we investigate the distribution of maternal pre-pregnancy BMI and educational attainment, both recorded in the GUSTO cohort study \citep{soh_2013_gusto}. Figure~\ref{fig:maternal_bmi_boxplot} shows the boxplots of maternal BMI before pregnancy, grouped by estimated cluster. While Cluster 1 (red) shows the highest variability, Cluster 2 (blue) shows a higher median. Cluster 3 (green) has the lowest median and variability. 

\begin{figure}[t!]
\centering
\includegraphics[width=.9\linewidth]{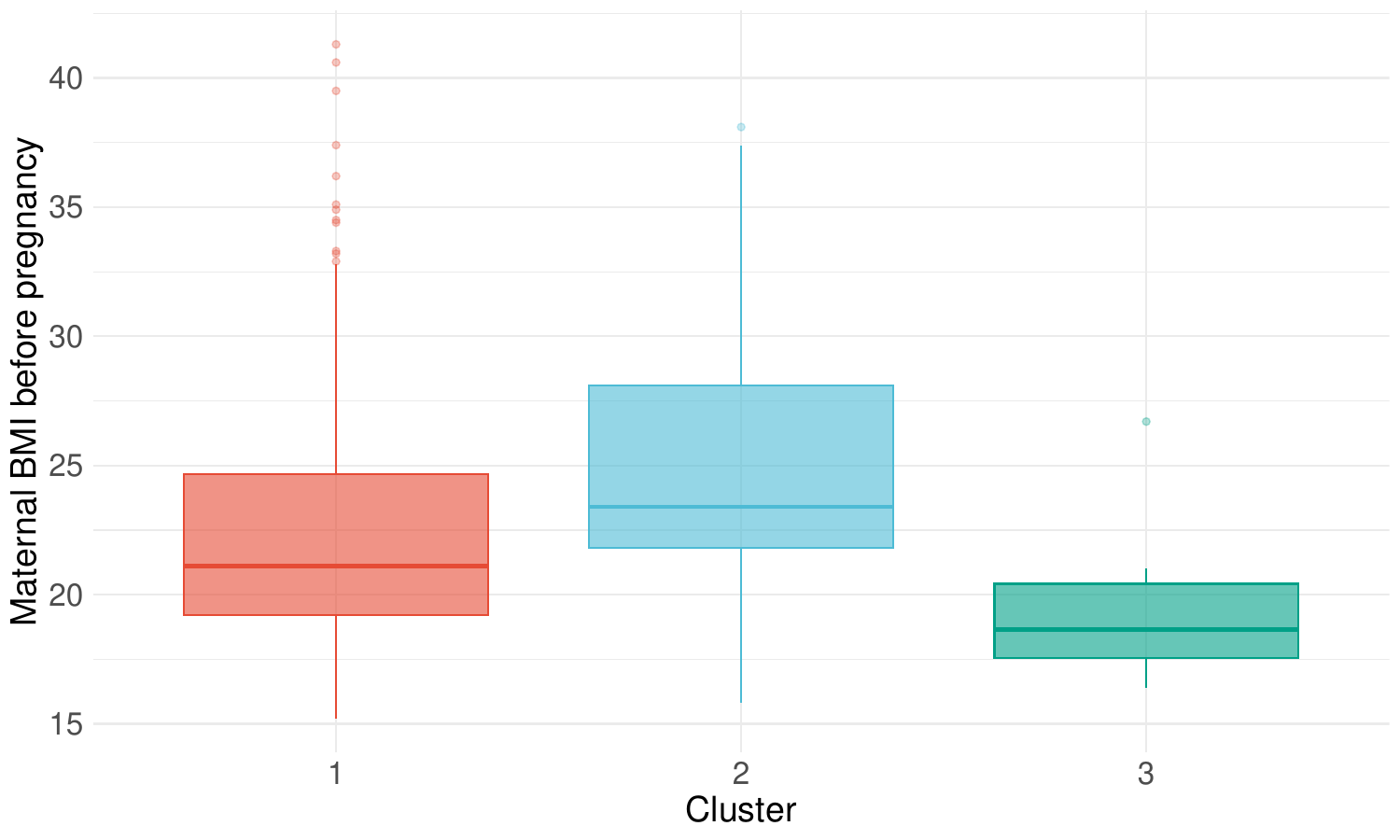}
\caption{GUSTO. Boxplots of pre-pregnancy maternal BMI across the three clusters estimated by minimising the Binder loss function.}
\label{fig:maternal_bmi_boxplot}
\end{figure}

\begin{figure}[t!]
\centering
\includegraphics[width=.9\linewidth]{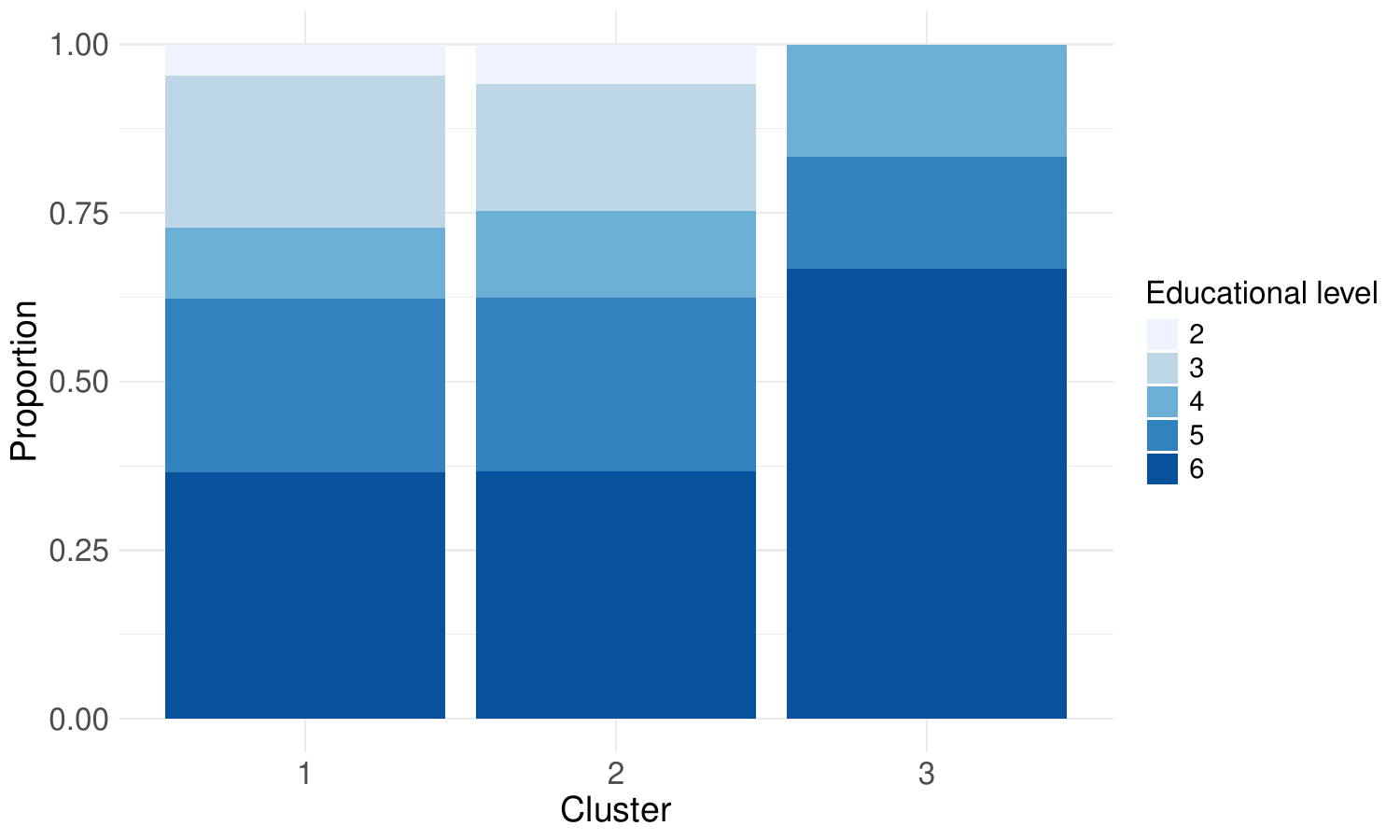}
\caption{GUSTO. Distribution of maternal education levels across the three clusters estimated by minimising the Binder loss function.}
\label{fig:maternal_education_stacked_barplot}
\end{figure}

Maternal education is measured on a 6-point ordinal scale, ranging from 1 (No education) to 6 (University degree, including Bachelors, Masters, or PhD). Intermediate categories reflect key stages of the Singaporean education system: Primary (PSLE), Secondary (GCE O/N Levels), vocational education (ITE/NITEC), and pre-university qualifications (GCE A Levels, Polytechnic, or Diploma). The distribution of educational levels across clusters is illustrated in Figure~\ref{fig:maternal_education_stacked_barplot}. Note that the level ``No education'' is not encountered in this dataset. The distribution of educational attainment in the first two clusters is nearly identical. The third cluster shows a greater proportion of mothers with the highest education level and no representation at levels 2 and 3; however, due to its small sample size the overall analysis indicates no significant differences in maternal education among the three clusters.


\end{document}